\documentclass[
  superscriptaddress,
  reprint,
  amsmath,
  amssymb,
  aps,
  floatfix,
  pre
]{revtex4-2}

\usepackage{graphicx}
\usepackage{dcolumn}
\usepackage{bm}
\usepackage{xcolor}
\usepackage{microtype}
\usepackage{amsthm}
\usepackage{tikz}
\usepackage{pgfplots}
\newcommand{\Pbb}{\mathbb{P}}
\newcommand{\Ebb}{\mathbb{E}}
\newcommand{\Vbb}{\mathbb{V}}
\newcommand{\Fbb}{\mathbb{F}}
\newenvironment{proposition}[1][]{%
  \par\addvspace{6pt}\noindent\textit{#1.}\ %
}{%
  \par\addvspace{6pt}%
}
\pgfplotsset{compat=1.18}
\definecolor{plotblue1}{HTML}{6BAED6}
\definecolor{plotblue2}{HTML}{3182BD}
\definecolor{plotblue3}{HTML}{35617D}
\definecolor{plotred}{HTML}{D62728}
\definecolor{plotorange1}{HTML}{FDAE6B}
\definecolor{plotorange2}{HTML}{F16913}
\definecolor{plotorange3}{HTML}{A63603}
\definecolor{plotgreen1}{HTML}{74C476}
\definecolor{plotgreen2}{HTML}{31A354}
\definecolor{plotgreen3}{HTML}{006D2C}
\definecolor{plotpurple1}{HTML}{BCBDDC}
\definecolor{plotpurple2}{HTML}{807DBA}
\definecolor{plotpurple3}{HTML}{54278F}
\pgfplotsset{
  publication axis/.style={
    axis line style={black, line width=0.65pt},
    tick style={black, line width=0.55pt},
    tick align=inside,
    major tick length=3pt,
    minor tick length=1.5pt,
    label style={font=\small},
    tick label style={font=\small},
    title style={font=\normalsize},
    legend cell align={left},
    legend style={
      font=\scriptsize,
      draw=black!20,
      fill=white,
      fill opacity=0.90,
      text opacity=1,
      rounded corners=1pt,
      inner xsep=3pt,
      inner ysep=2pt,
      row sep=-0.5pt
    },
    ymajorgrids=true,
    grid style={black!8, line width=0.35pt},
    axis background/.style={fill=white}
  }
}
\makeatletter
\patchcmd{\@makecaption}{\small\rmfamily}{\footnotesize\rmfamily}{}{}
\patchcmd{\@makecaption}{\flushing}{\raggedright}{}{}
\patchcmd{\@makecaption}
  {\hb@xt@\hsize{\hfil\unhbox\@tempboxa\hfil}}
  {\hb@xt@\hsize{\unhbox\@tempboxa\hfil}}
  {}{}
\makeatother

\newcommand{\tikzfigure}[2]{%
  \resizebox{#1}{!}{#2}%
}
\newcommand{\figurepanel}[4]{%
  \begin{minipage}{#1}%
    \raggedright\small\textbf{(#2)} #3\par\smallskip
    \centering
    \tikzfigure{\linewidth}{#4}%
  \end{minipage}%
}

\makeatletter
\def\label#1{\@bsphack
  \begingroup
  \UseHookWithArguments{label}{1}{#1}%
  \protected@write\@auxout{}%
    {\string\newlabel{#1}{{\@currentlabel}{\thepage}%
      {\@currentlabelname}{\@currentHref}{\@kernel@reserved@label@data}}}%
  \endgroup
  \@esphack}
\makeatother
\usepackage{hyperref}

\makeatletter
\patchcmd{\@bibdataout@aps}{author="08"}{author="48"}{}{}
\patchcmd{\@bibdataout@aps}{author="08"}{author="48"}{}{}
\makeatother

\begin{document}

\title{Transcription rate dynamics and RNA copy number noise: \\ General relations and data-driven predictions}

\author{Amara McCune}
\email{amaramccune@nyu.edu}
\affiliation{Center for Soft Matter Research, Department of Physics, New York University, New York, NY 10003, USA}
\author{Ido Golding}
\affiliation{Department of Physics, University of Illinois at Urbana-Champaign, Urbana, IL 61801, USA}
\affiliation{Department of Microbiology, University of Illinois at Urbana-Champaign, Urbana, IL 61801, USA}
\author{Shlomi Reuveni}
\affiliation{School of Chemistry, Tel Aviv University, Tel Aviv 6997801, Israel}
\affiliation{The Raymond and Beverly Sackler Center for Computational Molecular and Materials Science, Tel Aviv University, Tel Aviv 6997801, Israel}
\affiliation{The Center for Physics and Chemistry of Living Systems, Tel Aviv University, Tel Aviv 6997801, Israel}
\author{Sarah Kostinski}
\email{sk10775@nyu.edu}
\affiliation{Center for Soft Matter Research, Department of Physics, New York University, New York, NY 10003, USA}

\date{\today}

\begin{abstract}
  Stochastic models of gene expression typically begin with a microscopic model of transcription and propagate its statistics to the RNA distribution. Here we develop a doubly stochastic framework in which RNA production is Poisson conditional on a time-varying transcription rate $\lambda(t)$. For RNA lifetimes drawn independently from an arbitrary finite-mean distribution, we derive the RNA Fano factor in terms of the lifetime survival function and the autocovariance of $\lambda(t)$. For a Poisson degradation process with rate $\mu$, the survival kernel provides an exponential temporal filter, and the result depends only on the mean, variance, and normalized autocorrelation of $\lambda(t)$. These quantities may be calculated from an explicit rate model or, under ergodicity and adequate sampling, estimated from a sufficiently long rate trajectory. We validate the data-driven estimator using simulated transcription-rate trajectories, without supplying the known autocorrelation to the estimator. We also obtain exact analytical results for transcription-rate dynamics modeled by a discrete $\mathrm{M/M/1}$ process and by drift--diffusion with reflecting, periodic, or first passage reset boundaries, and verify each result by direct simulation of the coupled transcription-rate and RNA copy-number processes.
\end{abstract}

\maketitle

\section{Introduction}
\label{sec:introduction}

Transcription, the production of RNA from a gene, is an inherently stochastic process~\cite{golding2024primer,raj2008nature,golding2005real,mcadams1997stochastic}. It involves many probabilistic encounters in the cell, such as the binding of transcriptional activators and repressors to promoter and operator sites~\cite{bintu2005transcriptional,sanchez2011effect}. These random events lead to fluctuations in the number of molecules of a given RNA species present, which we refer to as the RNA copy number. Noise in the RNA copy number is typically quantified by statistical measures such as the Fano factor and squared coefficient of variation, both of which involve the mean and variance of the copy number distribution~\cite{paulsson2004summing,paulsson2005models,munsky2012using}. These normalized measures allow copy-number fluctuations to be compared across genes and conditions, and provide information about the underlying transcriptional process. 

Consider one RNA species in a single cell, and let $n(t)$ be its molecular copy number at time $t$.  New molecules appear as transcription is carried out, and each molecule is then degraded. Both processes are random. In addition, the transcription rate can change with the state of the gene, the cell, and its environment~\cite{swain2002intrinsic,paulsson2004summing}. The copy number therefore varies in time and across cells. Mechanistic, thermodynamic, and data-driven models have been developed to connect these processes to measured copy number statistics~\cite{thattai2001intrinsic,shahrezaei2008analytical,bintu2005transcriptional,sanchez2011effect,munsky2012using}, and to derive relations for the Fano factor and squared coefficient of variation~\cite{elowitz2002stochastic,swain2002intrinsic,paulsson2004summing,paulsson2005models,raj2008nature}.

A standard approach in stochastic gene-expression modeling is to specify a microscopic transcription model and derive the resulting distribution of $n(t)$~\cite{paulsson2005models,munsky2012using,sanchez2013genetic}. In the two-state telegraph model, for example, a gene switches between inactive and active states at fixed rates~\cite{peccoud1995markovian,kepler2001stochasticity,sanchez2013genetic,rieckh2014noise}. Experimental measurements, however, show that transcription rates can vary in time, over the cell cycle, and in response to upstream regulation~\cite{golding2005real,raj2006stochastic,chubb2006transcriptional,suter2011mammalian,larson2011real,larson2013direct,zopf2013cellcycle,padovan2015single,corrigan2016continuum,featherstone2016spatially,larsson2019genomic,rodriguez2020transcription,tunnacliffe2020what}. To incorporate such time dependence, here we treat the transcription rate $\lambda(t)$ as a nonnegative stochastic process. Conditional on a realization of $\lambda(t)$, RNA production is then a Poisson process with instantaneous rate $\lambda(t)$. The model is doubly stochastic because the random process $\lambda(t)$ drives a second random process of RNA production~\cite{cox1955statistical,dixit2013quantifying}. We use messenger RNA (mRNA) terminology because mRNA is the principal experimental application, but the derivations can be applied to any RNA species satisfying the stated production and lifetime assumptions.

In Section~\ref{sec:layered-framework} we derive the copy-number mean and Fano factor for mRNA production that is Poisson conditional on the complete transcription-rate trajectory $\lambda(\cdot)$, independent mRNA lifetimes, and a stationary transcription-rate process. For a general lifetime distribution, the result depends on its survival probability and on the mean, variance, and normalized autocorrelation of $\lambda(t)$. For exponential degradation with rate $\mu$, the survival probability is fixed, and the remaining inputs are $\Ebb[\lambda]$, $\Vbb[\lambda]$, and $\rho(h)$, which denote the expected value, variance, and normalized autocorrelation at lag time $h$, respectively. The resulting Fano factor is the sum of a Poisson baseline and a second term transmitted from fluctuations in $\lambda(t)$. A companion work uses the same framework to study how the copy-number variance scales with its mean~\cite{mccune2026subquadratic}; the present work develops and applies the framework itself.

The result can be used in two ways. First, a model for $\lambda(t)$ can be specified and its mean, variance, and autocorrelation calculated. Section~\ref{sec:existing-models} shows that constitutive expression, a random time-independent rate, the two-state telegraph model, and a cell-cycle null model are all recovered as special cases~\cite{zopf2013cellcycle,beentjes2020exact,pountain2024transcription}. Sections~\ref{sec:mm1-queue} and~\ref{sec:drift-diffusion} then derive results for four additional rate processes: a discrete $\mathrm{M/M/1}$ process and drift--diffusion with reflecting, periodic, or first-passage-reset boundaries.

These four rate processes represent different assumptions about the dynamics of $\lambda(t)$. The $\mathrm{M/M/1}$ model changes the rate in fixed increments and provides a tractable discrete example. The reflecting drift--diffusion process is its continuous counterpart on the nonnegative real line. The periodic model confines the rate to a finite interval and identifies its endpoints, producing recurrent dynamics without an explicit one-way reset. The first-passage model instead resets the rate to its lower boundary when it first reaches the upper boundary. This last rule can represent a rate that builds up during a cell cycle and drops at division~\cite{zopf2013cellcycle,padovan2015single,beentjes2020exact}. All four analytical results are compared with direct stochastic simulations.

Second, the same result can be applied directly to a measured trajectory of $\lambda(t)$~\cite{golding2005real,larson2011real,corrigan2016continuum,featherstone2016spatially}. If the trajectory is stationary and ergodic, its mean, variance, and autocorrelation can be estimated without choosing a parametric model. Section~\ref{sec:data-driven-recipe} gives the corresponding estimator, and Section~\ref{sec:data-driven-validation} tests it on synthetic data. This data-driven calculation also applies when the rate process is non-Markovian or its autocorrelation has no simple closed form.

Derivations of all analytical results in this work are provided in the appendices.

\section{Doubly Stochastic Model of mRNA Copy Number}
\label{sec:layered-framework}

\subsection{General lifetime distribution}
\label{subsec:layered-general-model}

We first consider the general case of a stochastic transcription rate and an arbitrary mRNA lifetime distribution. Note that time $t$ ranges from $-\infty$ to $+\infty$. Let $\lambda(t)$ denote the stochastic transcription rate at time $t$. The transcription rate is nonnegative at every time: $\lambda(t)\geq0$.  Conditional on a prescribed trajectory $\lambda(\cdot)$, mRNA production is a nonhomogeneous Poisson process with instantaneous rate $\lambda(t)$. Let $n(t)$ denote the number of mRNA molecules present at time $t$, where $n(t)$ is a nonnegative integer. 
Also, let us denote by $\tau$ the nonnegative lifetime of an individual mRNA and assume $\Ebb[\tau]<\infty$. We assume that each mRNA lifetime is drawn independently from the distribution of $\tau$, and that the lifetimes are independent of the production process and of $\lambda(\cdot)$. For $s\geq0$, we define the cumulative distribution function $\varphi(s)$ and survival function $\overline{\varphi}(s)$ of $\tau$ by
\begin{subequations}
  \label{eq:layered-lifetime-functions}
  \begin{align}
    \varphi(s)
    &\mathrel{:=} \Pbb(\tau\leq s),
    \qquad 0\leq\varphi(s)\leq1,
    \label{eq:layered-lifetime-cdf}\\
    \overline{\varphi}(s)
    &\mathrel{:=} \Pbb(\tau>s)
    =1-\varphi(s),
    \qquad 0\leq\overline{\varphi}(s)\leq1.
    \label{eq:layered-lifetime-survival}
  \end{align}
\end{subequations}

\begin{proposition}[Conditional copy-number distribution]
For an arbitrary mRNA lifetime distribution, the expected copy number conditional on the entire transcription-rate trajectory is
\begin{align}
  \Ebb\!\left[n(t)\,\middle|\,\lambda(\cdot)\right]
  &= \int_0^\infty
  \overline{\varphi}(s)\,
  \lambda(t-s)\,d s \notag\\
  &= \bigl(\lambda*\overline{\varphi}\bigr)(t)\,,
  \label{eq:newell}
\end{align}
where $*$ denotes the causal convolution written explicitly in the preceding line. Moreover,
\begin{equation}
  n(t)\mid\lambda(\cdot)
  \sim
  \operatorname{Poisson}\,\bigl(
    \Ebb[n(t)\mid\lambda(\cdot)]
  \bigr)\,.
  \label{eq:layered-conditional-poisson}
\end{equation}
\end{proposition}
\noindent Thus, even after the entire trajectory $\lambda(\cdot)$ is prescribed, the random production times and mRNA lifetimes cause $n(t)$ to fluctuate.

\subsection{\texorpdfstring{Stationary transcription rate: mean and variance of the copy number}{Stationary transcription rate: mean and variance of the copy number}}
\label{subsec:layered-stationary-statistics}

We now assume that $\lambda(t)$ is stationary, with a finite positive mean and finite variance.  The mRNA lifetimes remain independent of $\lambda(\cdot)$, as assumed above.
\begin{proposition}[Stationary mean]
Under this stationarity assumption on $\lambda(t)$, the induced copy-number distribution is also stationary. Its mean is 
\begin{equation}
  \Ebb[n]
  = \Ebb[\lambda]\Ebb[\tau]\,.
  \label{eq:layered-stationary-mean}
\end{equation}
\end{proposition}
\noindent The stationary mean therefore depends on the mRNA lifetime distribution only through $\Ebb[\tau]$. The stationary variance, on the other hand, generally depends on the complete lifetime distribution.

For a nonnegative integer-valued random variable $x$ with finite positive mean and finite variance, we define its Fano factor by 
\begin{equation}
  \Fbb[x] \mathrel{:=} \frac{\Vbb[x]}{\Ebb[x]}\,.
  \label{eq:layered-fano-definition}
\end{equation}
Provided that $0<\Vbb[\lambda]<\infty$, we denote the normalized autocorrelation of the transcription rate at lag $h$ by 
\begin{equation}
  \rho(h)
  \mathrel{:=}
  \frac{\Ebb[\lambda(t+h)\lambda(t)]-\Ebb[\lambda]^2}
       {\Vbb[\lambda]}\,.
  \label{eq:layered-autocorrelation-definition}
\end{equation}

\begin{proposition}[Fano factor for arbitrary mRNA lifetimes and a stationary transcription rate]
If $0<\Vbb[\lambda]<\infty$, then
\begin{align}
  \Fbb[n] ={}& 1 \mathbin{+}{} \notag\\
  &\frac{\Vbb[\lambda]}{\Ebb[n]}
  \int_0^\infty\!\!\int_0^\infty
  \overline{\varphi}(s_1)\,
  \overline{\varphi}(s_2)
  \rho(|s_2-s_1|)
  \,d s_1\,d s_2\,.
  \label{eq:general-fano}
\end{align}
\end{proposition}

\noindent The first term $1$ on the right-hand side is the intrinsic Poisson contribution from the random production times and mRNA lifetimes. The second term is the additional copy-number variability transmitted from fluctuations in $\lambda(t)$. The two survival factors account for whether mRNAs of different ages remain present, while $\rho(|s_2-s_1|)$ measures how strongly the corresponding transcription rates are correlated.

\subsection{Exponential lifetimes}
\label{subsec:layered-exponential-lifetimes}

We now specialize to the exponential survival function
\begin{equation}
  \overline{\varphi}(s)
  =
  e^{-\mu s},
  \qquad \mu>0.
  \label{eq:layered-exponential-survival}
\end{equation}
where $\mu$ is the degradation rate. This law is memoryless: every surviving mRNA has the same instantaneous degradation rate $\mu$, independent of its age.  Furthermore, it is assumed that each mRNA molecule degrades independently. Hence their total instantaneous degradation rate is $\mu n(t)$. The stationary mean is then
\begin{equation}
  \Ebb[n] = \frac{\Ebb[\lambda]}{\mu}\,.
  \label{eq:layered-exponential-mean}
\end{equation}

\begin{proposition}[Fano factor for exponential lifetimes]
Under the preceding finite-moment assumptions, 
\begin{equation}
  \Fbb[n]
  =
  1+
  \frac{\Vbb[\lambda]}{\mu\,\Ebb[n]}
  \int_0^\infty e^{-\mu h}\rho(h)\,d h\,.
  \label{eq:poisson-fano}
\end{equation}
\end{proposition}

\begin{proposition}[Poisson floor and slow-fluctuation ceiling]
From the expression above, we obtain bounds on the Fano factor. For every stationary transcription-rate process satisfying the assumptions above,
\begin{equation}
  1
  \leq
  \Fbb[n]
  \leq
  1+
  \frac{\Vbb[\lambda]}{\mu\,\Ebb[\lambda]}.
  \label{eq:static-ceiling}
\end{equation}
\end{proposition}

\noindent The lower bound is the intrinsic Poisson floor of the model. The Fano factor is computed from the stationary distribution of $n$. Experimentally, this distribution may be estimated from a snapshot of many identically prepared cells. If the stationary process is also ergodic, the same distribution may instead be estimated from one sufficiently long time trace. When $\lambda(t)$ changes negligibly during an mRNA lifetime, conditional on the sampled value of $\lambda$, the mRNA count has the approximate distribution $n\mid\lambda\sim\operatorname{Poisson}(\lambda/\mu)$. The law of total variance~\cite{ross2014introduction} then gives the within-sample Poisson variance plus the between-sample variance of $\lambda/\mu$, yielding the upper bound. Because this bound is approached when rate fluctuations are slow relative to mRNA degradation, we call it the slow-fluctuation ceiling. The later sections obtain explicit noise predictions by specifying the dynamics of $\lambda(t)$ and calculating $\Ebb[\lambda]$, $\Vbb[\lambda]$, and $\rho(h)$.

\section{Existing Models as Special Cases}
\label{sec:existing-models}
Equation~\eqref{eq:general-fano} converts the statistics of the fluctuating transcription rate $\lambda(t)$ into a prediction for the mRNA Fano factor. We first demonstrate this procedure using familiar models for which the mRNA statistics are already known. Specifically, these models involve a rate fixed at the same value in every realization (constitutive Poisson expression)~\cite{thattai2001intrinsic}; a rate that is time-independent within each realization but random across realizations (mixed-Poisson distribution)~\cite{cox1955statistical}; and a rate that switches between two states (standard telegraph model)~\cite{peccoud1995markovian,kepler2001stochasticity}. We also consider the cell-cycle dosage model, in which the transcription rate halves at cell division. For this model, we assume that the mRNA lifetime is short relative to the cell-cycle duration and derive the resulting slow-cell-cycle approximation~\cite{zopf2013cellcycle,beentjes2020exact,pountain2024transcription}. Reproducing these known results provides a check of the framework before we apply the same procedure to the $\mathrm{M/M/1}$ and drift--diffusion models in Sections~\ref{sec:mm1-queue} and~\ref{sec:drift-diffusion}.

\subsection{Constitutive expression}

If $\lambda(t)\mathrel{:=}\lambda_0$, where $\lambda_0>0$ is fixed and identical across all realizations, then
\begin{equation}
  \Ebb[\lambda]=\lambda_0,
  \qquad
  \Vbb[\lambda]=0.
  \label{eq:constitutive-rate-moments}
\end{equation}
Because the rate trajectory is constant-valued and identical in every realization, the conditional mean is simply
\begin{equation}
  \Ebb[n\mid\lambda(\cdot)]
  =\lambda_0\int_0^\infty
  \overline{\varphi}(s)\,d s
  =\lambda_0\Ebb[\tau]
  =\Ebb[n].
  \label{eq:constitutive-copy-number-mean}
\end{equation}
The law of total variance therefore gives
\begin{subequations}
\label{eq:constitutive-fano}
\begin{align}
  \Vbb[n]
  &=\Ebb\!\left[\Vbb[n\mid\lambda(\cdot)]\right]
  +\Vbb\!\left[\Ebb[n\mid\lambda(\cdot)]\right]
  =\Ebb[n],
  \label{eq:constitutive-variance}\\
  \Fbb[n]
  &=\frac{\Vbb[n]}{\Ebb[n]}=1.
  \label{eq:constitutive-fano-result}
\end{align}
\end{subequations}
We call this constitutive expression because the gene has no switching regulation and produces mRNA at a constant underlying rate, although individual production events remain random. This is ordinary Poisson expression and is the minimum Fano factor achievable within the present conditionally Poisson production model. Deviations above this floor require either stochastic fluctuations in $\lambda(t)$ or additional sources of variation.

\subsection{Random time-independent transcription rate}

Let $\Lambda_0$ be a nonnegative random variable with finite positive mean and finite variance, drawn independently for each realization, and set  
\begin{equation}
  \lambda(t)\mathrel{:=}\Lambda_0,
  \qquad \text{for all }t.
  \label{eq:random-static-rate}
\end{equation}
This idealized model represents persistent cell-to-cell heterogeneity that is effectively fixed over the observation interval. It is also the limiting case in which rate fluctuations are much slower than mRNA degradation.  The rate varies across realizations but not over time within a realization. Consequently,
\begin{equation}
  \Ebb[\lambda]=\Ebb[\Lambda_0],
  \qquad
  \Vbb[\lambda]=\Vbb[\Lambda_0],
  \qquad
  \rho(h)=1.
  \label{eq:random-static-rate-statistics}
\end{equation}
Conditional on $\Lambda_0$, the stationary copy number has a Poisson distribution:
\begin{equation}
  n\mid\Lambda_0
  \sim
  \operatorname{Poisson}\,\bigl(
    \Lambda_0\Ebb[\tau]
  \bigr).
  \label{eq:random-static-conditional-copy-number}
\end{equation}
The laws of total expectation and total variance~\cite{ross2014introduction} give
\begin{subequations}
\label{eq:random-static-copy-number-statistics}
\begin{align}
  \Ebb[n]
  &=\Ebb[\Lambda_0]\Ebb[\tau],
  \label{eq:random-static-copy-number-mean}\\
  \Vbb[n]
  &=\Ebb[n]
  +\Vbb[\Lambda_0]\Ebb[\tau]^2,
  \label{eq:random-static-copy-number-variance}\\
  \Fbb[n]
  &=1+
  \frac{\Vbb[\Lambda_0]}
       {\Ebb[\Lambda_0]}
  \Ebb[\tau].
  \label{eq:random-static-copy-number-fano}
\end{align}
\end{subequations}
Thus the unconditional copy number follows a mixed-Poisson distribution rather than a Poisson distribution. For exponential degradation, $\Ebb[\tau]=1/\mu$, and the Fano factor becomes
\begin{equation}
  \Fbb[n]
  =
  1+
  \frac{\Vbb[\Lambda_0]}
       {\mu\,\Ebb[\Lambda_0]},
  \label{eq:random-static-exponential-fano}
\end{equation}
which is the slow-fluctuation ceiling in equation~\eqref{eq:static-ceiling}. 

\subsection{Two-state telegraph model}

In the two-state telegraph model~\cite{peccoud1995markovian,kepler2001stochasticity}, let $x(t)$ denote the activity state of the gene at time $t$. For every real $t$, $x(t)$ is either $0$ or $1$: $x(t)=0$ means that the gene is inactive, whereas $x(t)=1$ means that it is active. The state switches randomly between these two values. For a short time interval of length $\Delta t$,
\begin{subequations}
\label{eq:telegraph-transitions}
\begin{align}
  \Pbb\bigl(x(t+\Delta t)=1\mid x(t)=0\bigr)
  &=k_{\mathrm{on}}\,\Delta t+o(\Delta t),
  \label{eq:telegraph-on-transition}\\
  \Pbb\bigl(x(t+\Delta t)=0\mid x(t)=1\bigr)
  &=k_{\mathrm{off}}\,\Delta t+o(\Delta t),
  \label{eq:telegraph-off-transition}
\end{align}
\end{subequations}
where $k_{\mathrm{on}}>0$ is the rate at which an inactive gene switches on, and $k_{\mathrm{off}}>0$ is the rate at which an active gene switches off, with units of inverse time. The corresponding transcription rate is then
\begin{equation}
  \lambda(t)
  =\lambda_{\mathrm{on}}x(t)
  =
  \begin{cases}
    0, & x(t)=0,\\
    \lambda_{\mathrm{on}}, & x(t)=1,
  \end{cases}
  \label{eq:telegraph-transcription-rate}
\end{equation}
where $\lambda_{\mathrm{on}}>0$ is the transcription rate while the gene is active, with units of mRNA produced per unit time. Thus $\lambda_{\mathrm{on}}$ determines how quickly mRNA is produced in the active state, whereas $k_{\mathrm{on}}$ and $k_{\mathrm{off}}$ determine how quickly the gene switches between its two states. We define
\begin{equation}
  \kappa\mathrel{:=}k_{\mathrm{on}}+k_{\mathrm{off}},
  \label{eq:telegraph-relaxation-rate}
\end{equation}
where $\kappa$ is the relaxation rate of the telegraph process and $1/\kappa$ is its correlation time. 

At stationarity, probability balance requires $\Pbb(x=0)k_{\mathrm{on}}=\Pbb(x=1)k_{\mathrm{off}}$. Together with $\Pbb(x=0)+\Pbb(x=1)=1$, this gives 
the stationary distribution of the transcription rate:
\begin{subequations}
\label{eq:telegraph-stationary-distribution}
\begin{align}
  \Pbb(x=1)
  &=\Pbb(\lambda=\lambda_{\mathrm{on}})
  =p_{\mathrm{on}}
  \mathrel{:=}\frac{k_{\mathrm{on}}}{\kappa},
  \label{eq:telegraph-on-probability}\\
  \Pbb(x=0)
  &=\Pbb(\lambda=0)
  =1-p_{\mathrm{on}}
  =\frac{k_{\mathrm{off}}}{\kappa}.
  \label{eq:telegraph-off-probability}
\end{align}
\end{subequations}
\noindent
Its stationary mean, variance, and normalized autocorrelation are therefore
\begin{subequations}
\label{eq:telegraph-rate-statistics}
\begin{align}
  \Ebb[\lambda]
  &=\lambda_{\mathrm{on}}p_{\mathrm{on}},
  \label{eq:telegraph-rate-mean}\\
  \Vbb[\lambda]
  &=\lambda_{\mathrm{on}}^2
  p_{\mathrm{on}}(1-p_{\mathrm{on}}),
  \label{eq:telegraph-rate-variance}\\
  \rho(h)
  &=e^{-\kappa|h|},
  \qquad .
  \label{eq:telegraph-rate-autocorrelation}
\end{align}
\end{subequations}
\noindent
For exponential degradation,
\begin{equation}
  \Ebb[n]
  =\frac{\Ebb[\lambda]}{\mu}
  =\frac{\lambda_{\mathrm{on}}p_{\mathrm{on}}}{\mu}.
  \label{eq:telegraph-copy-number-mean}
\end{equation}
Substitution into equation~\eqref{eq:poisson-fano} gives
\begin{align}
  \Fbb[n]
  &=1+
  \frac{\Vbb[\lambda]}{\mu\,\Ebb[n]}
  \int_0^\infty e^{-(\mu+\kappa)h}\,d h \notag\\
  &=1+
  \frac{\Vbb[\lambda]}
       {\mu(\mu+\kappa)\Ebb[n]} \notag\\
  &=1+
  \frac{\lambda_{\mathrm{on}}(1-p_{\mathrm{on}})}
       {\mu+\kappa}.
  \label{eq:telegraph-fano}
\end{align}
Equivalently, the excess above the Poisson floor can be written as
\begin{equation}
  \Fbb[n]
  =1+
  \frac{\Vbb[\lambda]}
       {\mu\,\Ebb[\lambda]}
  \frac{\mu}{\mu+\kappa}.
  \label{eq:telegraph-filtering}
\end{equation}
\noindent The first factor on the right-hand side is the excess noise that would be obtained if the transcription rate remained fixed during an mRNA lifetime. The second factor describes temporal filtering by mRNA degradation. To compare the two processes explicitly, let us define the switching correlation time $\tau_{\mathrm{switch}}:=1/\kappa$ and the mRNA degradation time $\tau_{\mathrm{mRNA}}:=1/\mu$.  Holding $p_{\mathrm{on}}$ and $\lambda_{\mathrm{on}}$ fixed, rapid gene switching, $\tau_{\mathrm{switch}}\ll\tau_{\mathrm{mRNA}}$ (equivalently $\kappa\gg\mu$), averages the active and inactive states over an mRNA lifetime and gives $\Fbb[n]\to1$. Slow switching, $\tau_{\mathrm{switch}}\gg\tau_{\mathrm{mRNA}}$ (equivalently $\kappa\ll\mu$), leaves the transcription rate nearly fixed over an mRNA lifetime and gives
\begin{equation}
  \Fbb[n]
  \longrightarrow
  1+
  \frac{\Vbb[\lambda]}
       {\mu\,\Ebb[\lambda]}
  =
  1+
  \frac{\lambda_{\mathrm{on}}(1-p_{\mathrm{on}})}{\mu},
  \label{eq:telegraph-static-limit}
\end{equation}
which is the slow-fluctuation ceiling in equation~\eqref{eq:static-ceiling}. Thus the general formula recovers the standard telegraph-model result~\cite{peccoud1995markovian,kepler2001stochasticity} and shows explicitly how the switching timescale controls the additional mRNA noise.

\subsection{Cell-cycle null model}

The cell-cycle null model isolates the copy-number variation caused by gene replication, without introducing any additional transcriptional regulation~\cite{pountain2024transcription}. Let $a\in[0,c)$ denote the age of a cell since its most recent division, where $c>0$ is the cell-cycle duration. Suppose that the gene replicates at age $\theta c$, where $\theta\in(0,1)$ is the fraction of the cell cycle completed at replication, and that each gene copy produces mRNA at the constant rate $\beta>0$. The corresponding transcription rate at cell age $a$ is
\begin{equation}
  \lambda(a)
  =
  \begin{cases}
    \beta, & 0\leq a<\theta c,\\
    2\beta, & \theta c\leq a<c.
  \end{cases}
  \label{eq:cell-cycle-rate}
\end{equation}
Thus replication doubles the transcription rate, and division returns each daughter cell to the one-copy state. This null model isolates the discrete dosage change and deliberately omits continuous cell-growth effects on transcription, which are treated in growth-coupled gene-expression models~\cite{lin2018homeostasis,golding2024primer}.

To obtain the distribution seen in a population snapshot, assume an asynchronous, exponentially growing population in which every cell divides at age $c$. The cell-age density $p_{\mathrm{age}}(a)$ is the probability density of the age of a cell selected from a population snapshot.  Exponential growth overrepresents young cells, giving~\cite{powell1956growth,beentjes2020exact}.
\begin{equation}
  p_{\mathrm{age}}(a)
  \mathrel{:=}
  \frac{2\ln 2}{c}\,2^{-\frac{a}{c}},
  \qquad 0\leq a<c.
  \label{eq:cell-cycle-age-density}
\end{equation}
This density is normalized because
\begin{equation*}
  \int_0^c p_{\mathrm{age}}(a)\,d a=1.
\end{equation*}
Consequently, the fraction of sampled cells that have already passed the replication time is
\begin{equation}
  f
  \mathrel{:=}
  \Pbb(a\geq\theta c)
  =\int_{\theta c}^{c}p_{\mathrm{age}}(a)\,d a
  =2^{1-\theta}-1.
  \label{eq:cell-cycle-replicated-fraction}
\end{equation}
The snapshot transcription rate therefore has the binary distribution
\begin{equation}
  \Pbb(\lambda=\beta)=1-f,
  \qquad
  \Pbb(\lambda=2\beta)=f,
  \label{eq:cell-cycle-rate-distribution}
\end{equation}
and hence
\begin{subequations}
\label{eq:cell-cycle-moments}
\begin{align}
  \Ebb[\lambda]
  &=\beta(1+f),
  \label{eq:cell-cycle-rate-mean}\\
  \Vbb[\lambda]
  &=\beta^2 f(1-f).
  \label{eq:cell-cycle-rate-variance}
\end{align}
\end{subequations}

For exponential degradation, assume that the cell cycle is slow compared with an mRNA lifetime, $\mu c\gg1$. The transcription rate can then be treated as fixed at its current value over the recent history that contributes appreciably to the present copy number. Under this slow-fluctuation approximation,
\begin{equation}
  n\mid\lambda
  \sim
  \operatorname{Poisson}\!\left(\frac{\lambda}{\mu}\right).
  \label{eq:cell-cycle-conditional-poisson}
\end{equation}
The laws of total expectation and total variance now give
\begin{subequations}
\label{eq:cell-cycle-copy-number-statistics}
\begin{align}
  \Ebb[n]
  &=\frac{\Ebb[\lambda]}{\mu}
  =\frac{\beta(1+f)}{\mu},
  \label{eq:cell-cycle-copy-number-mean}\\
  \Vbb[n]
  &\approx\Ebb[n]
  +\frac{\beta^2 f(1-f)}{\mu^2},
  \label{eq:cell-cycle-copy-number-variance}\\
  \Fbb[n]
  &\approx1+
  \frac{\beta^2 f(1-f)}{\mu^2\Ebb[n]}
  =1+
  \frac{\beta f(1-f)}{\mu(1+f)}.
  \label{eq:cell-cycle-fano}
\end{align}
\end{subequations}
The degradation kernel in equation~\eqref{eq:poisson-fano} gives appreciable weight only to lags $h$ of order $1/\mu$. When $\mu c\gg1$, these lags are much shorter than a cell cycle, so the transcription rate changes negligibly over them and $\rho(h)\approx1$. The Fano factor consequently approaches the slow-fluctuation ceiling~\eqref{eq:static-ceiling}.  The excess above the Poisson floor is therefore caused by mixing cells with transcription rates $\beta$ and $2\beta$ in the snapshot. It vanishes when all sampled cells have the same transcription rate, $f=0$ or $f=1$. If the mRNA lifetime is not short relative to the cell cycle, the static-rate approximation is insufficient and the full cell-cycle history of $\lambda(t)$ must be retained.


\section{General Recipe and Data-Driven Application}
\label{sec:data-driven-recipe}
In the preceding section, we evaluated equation~\eqref{eq:poisson-fano} by choosing explicit models for the transcription-rate process $\lambda(t)$ and deriving their statistics. We now consider a data-driven alternative in which $\lambda(t)$ is stationary, and one of its trajectories is measured directly. To infer ensemble statistics from that trajectory, we additionally assume ergodicity. Under exponential degradation with known rate $\mu$, equation~\eqref{eq:poisson-fano} requires the mean $\Ebb[\lambda]$, variance $\Vbb[\lambda]$, and normalized autocorrelation $\rho(h)$. Ergodicity allows these ensemble quantities to be estimated from a sufficiently long trajectory. We evaluate the remaining integral numerically, thereby predicting the mRNA Fano factor without specifying a parametric model for the dynamics of $\lambda(t)$. Hatted symbols denote quantities estimated from the trajectory.

Suppose a stationary transcription-rate trajectory is observed at $N$ uniformly spaced times. The sampling times and total observation time are
\begin{equation}
  \begin{aligned}
    t_i &\mathrel{:=} t_1+(i-1)\Delta t,
    &&i=1,\ldots,N,\\
    T &\mathrel{:=} t_N-t_1=(N-1)\Delta t.
  \end{aligned}
  \label{eq:trajectory-sampling-times}
\end{equation}
\noindent
Estimate the stationary mean and variance from these observations as
\begin{subequations}
  \label{eq:trajectory-moment-estimates}
  \begin{align}
    \widehat{\Ebb}[\lambda]
    &\mathrel{:=} \frac{1}{N}\sum_{i=1}^{N}\lambda(t_i),
    \label{eq:trajectory-mean-estimate}\\
    \widehat{\Vbb}[\lambda]
    &\mathrel{:=} \frac{1}{N}\sum_{i=1}^{N}
    \left[\lambda(t_i)-\widehat{\Ebb}[\lambda]\right]^2.
    \label{eq:trajectory-variance-estimate}
  \end{align}
\end{subequations}
\noindent
For the lag $h_k\mathrel{:=} k\Delta t$, where $k=0,1,\ldots,K_{\max}$ and $1\leq K_{\max}<N$, we estimate the normalized autocorrelation by
\begin{align}
  \widehat{\rho}(h_k)
  \mathrel{:=}{}& \frac{1}{(N-k)\,\widehat{\Vbb}[\lambda]} \notag\\
  &\sum_{i=1}^{N-k}
  \bigl[\lambda(t_i)-\widehat{\Ebb}[\lambda]\bigr]
  \bigl[\lambda(t_{i+k})-\widehat{\Ebb}[\lambda]\bigr].
  \label{eq:rho-hat}
\end{align}
\noindent
Finally, we approximate the remaining integral by connecting adjacent sampled values with straight lines. This trapezoidal rule gives half weight to the two endpoints and full weight to the interior points, producing the Fano-factor estimate 
\begin{align}
  \widehat{\Fbb}[n]
  \mathrel{:=} 1+
  \frac{\widehat{\Vbb}[\lambda]}{\widehat{\Ebb}[\lambda]}\,\Delta t
  \Biggl[
  \frac{\widehat{\rho}(0)}{2}
  &+\sum_{k=1}^{K_{\max}-1} e^{-\mu h_k}\widehat{\rho}(h_k) \notag\\
  &+\frac{e^{-\mu h_{K_{\max}}}
  \widehat{\rho}(h_{K_{\max}})}{2}
  \Biggr].
  \label{eq:data-driven-fano}
\end{align}
The cutoff $K_{\max}$ is chosen so that the omitted, exponentially weighted correlation tail is negligible, while $\Delta t$ must be small enough to resolve the weighted autocorrelation. Equation~\eqref{eq:layered-exponential-mean} gives the corresponding mean-copy-number estimate,
\begin{equation}
  \widehat{\Ebb}[n]
  \mathrel{:=}
  \frac{\widehat{\Ebb}[\lambda]}{\mu}.
  \label{eq:data-driven-mean}
\end{equation}
Thus only the degradation rate $\mu$ must be known independently, for example from an mRNA half-life measurement.

Within the doubly stochastic framework and under the stationarity, ergodicity, and exponential-degradation assumptions above, the recipe does not require a parametric model for the dynamics of $\lambda(t)$. It applies whenever $\lambda(t)$ can be measured or inferred---for example, from MS2/PP7 tagging, intronic FISH, or live-cell fluorescence reporters of nascent transcription~\cite{golding2005real,larson2011real,raj2008imaging}. The next section demonstrates this process on synthetic data.

\section{Proof-of-Principle of Data-Driven Application}
\label{sec:data-driven-validation}
We test the data-driven recipe using a synthetic time series for which an analytical Fano-factor benchmark is known. Let $\lambda(t)$ follow the stationary Ornstein--Uhlenbeck process~\cite{uhlenbeck1930theory,gardiner2009handbook},
\begin{align}
  d\lambda(t) = -\gamma\bigl[\lambda(t)-\lambda_0\bigr]\,dt
  + \sigma\sqrt{2\gamma}\,dW(t),
  \label{eq:ou-process}
\end{align}
where $W(t)$ is a standard Wiener process. The Ornstein--Uhlenbeck process is a linear mean-reverting process driven by Gaussian noise. Its stationary distribution is therefore Gaussian, with mean $\lambda_0$ and variance $\sigma^2$.  At stationarity, $\Ebb[\lambda]=\lambda_0$, $\Vbb[\lambda]=\sigma^2$, and $\rho(h)=e^{-\gamma|h|}$, with $\gamma$ the relaxation rate. Defining $\kappa \mathrel{:=} \gamma/\mu$, equation~\eqref{eq:poisson-fano} becomes
\begin{align}
  \Fbb^{\mathrm{OU}}[n] = 1
  + \Ebb[n]\,\frac{\sigma^2}{\lambda_0^2}\,\frac{1}{1+\kappa},
  \label{eq:ou-fano}
\end{align}
which provides the OU benchmark. We take $\lambda_0 = 5$, $\sigma = 1$, $\gamma = 0.5$, and $\mu = 1$, giving $\Ebb[n] = 5$, $\kappa = 0.5$, and $\Fbb^{\mathrm{OU}}[n] = 17/15 \simeq 1.1333$. A Gaussian OU process is not strictly nonnegative, but here $\Pbb(\lambda<0)\simeq2.9\times10^{-7}$. When generating mRNA, we replace these rare negative values by zero; the resulting correction is negligible at the reported precision. We then simulate the coupled process $(\lambda(t),n(t))$ and obtain two estimates of $\Fbb[n]$ from each realization.

The data-driven estimate uses the sampled rate trajectory and the known degradation rate. We sample $\lambda(t)$ at intervals $\Delta t=0.1/\mu$, estimate its mean, variance, and autocorrelation using equations~\eqref{eq:trajectory-moment-estimates} and~\eqref{eq:rho-hat}, and calculate $\widehat{\Fbb}[n]$ from equation~\eqref{eq:data-driven-fano}. No analytical expression for $\rho(h)$ is supplied to this calculation. For comparison, the simultaneously simulated copy-number trajectory gives the direct estimate $\widehat{\Fbb}^{\mathrm{direct}}[n]\mathrel{:=}\widehat{\Vbb}[n]/\widehat{\Ebb}[n]$. Both estimates are compared with the OU benchmark in equation~\eqref{eq:ou-fano}.

For $100$ independent realizations of duration $T=10^4/\mu$, the data-driven estimate is $\widehat{\Fbb}[n]=1.1337\pm0.0003$, whereas the direct estimate is $\widehat{\Fbb}^{\mathrm{direct}}[n]=1.1360\pm0.0019$. The uncertainties are standard errors across realizations. Both estimates agree with the OU benchmark $17/15$ to within two standard errors (Fig.~\ref{fig:data-driven-validation}). The estimated autocorrelation agrees with $e^{-\gamma|h|}$ over its principal decay range, with finite-sample deviations after the correlation has become small. Thus the sampled rate trajectory recovers the analytical benchmark without assuming the functional form of $\rho(h)$. The same procedure applies when a measured transcription-rate trajectory has no closed-form autocorrelation.

\begin{figure*}[!tbp]
  \centering
  \figurepanel{0.48\textwidth}{a}{Estimated autocorrelation}%
    {%
\begin{tikzpicture}
\begin{axis}[
  publication axis,
  width=8.5cm,
  height=5.0cm,
  xlabel={lag $h$ (units of $1/\mu$)},
  ylabel={$\rho(h)$},
  xmin=0, xmax=10,
  ymin=-0.05, ymax=1.05,
  xtick={0,2,4,6,8,10},
  ytick={0,0.25,0.5,0.75,1},
  xmajorgrids=true,
  legend style={
    at={(0.5,1.02)},
    anchor=south,
    legend columns=-1,
    draw=none,
    fill=none,
    /tikz/every even column/.append style={column sep=5pt}
  }
]
\addplot[black, very thick, no marks] coordinates {(0,1) (0.0558659218,0.972453557) (0.111731844,0.945665921) (0.167597765,0.919616188) (0.223463687,0.894284033) (0.279329609,0.869649689) (0.335195531,0.845693934) (0.391061453,0.822398074) (0.446927374,0.799743932) (0.502793296,0.777713832) (0.558659218,0.756290582) (0.61452514,0.735457467) (0.670391061,0.715198229) (0.726256983,0.695497062) (0.782122905,0.676338592) (0.837988827,0.65770787) (0.893854749,0.639590357) (0.94972067,0.621971918) (1.00558659,0.604838804) (1.06145251,0.588177646) (1.11731844,0.571975444) (1.17318436,0.556219555) (1.22905028,0.540897685) (1.2849162,0.525997878) (1.34078212,0.511508507) (1.39664804,0.497418267) (1.45251397,0.483716164) (1.50837989,0.470391504) (1.56424581,0.457433891) (1.62011173,0.444833215) (1.67597765,0.432579642) (1.73184358,0.420663611) (1.7877095,0.409075825) (1.84357542,0.397807241) (1.89944134,0.386849067) (1.95530726,0.376192751) (2.01117318,0.365829979) (2.06703911,0.355752664) (2.12290503,0.345952944) (2.17877095,0.336423171) (2.23463687,0.327155909) (2.29050279,0.318143927) (2.34636872,0.309380194) (2.40223464,0.30085787) (2.45810056,0.292570306) (2.51396648,0.284511035) (2.5698324,0.276673768) (2.62569832,0.269052389) (2.68156425,0.261640953) (2.73743017,0.254433676) (2.79329609,0.247424933) (2.84916201,0.240609256) (2.90502793,0.233981327) (2.96089385,0.227535974) (3.01675978,0.221268167) (3.0726257,0.215173016) (3.12849162,0.209245765) (3.18435754,0.203481788) (3.24022346,0.197876589) (3.29608939,0.192425793) (3.35195531,0.187125146) (3.40782123,0.181970514) (3.46368715,0.176957874) (3.51955307,0.172083314) (3.57541899,0.167343031) (3.63128492,0.162733325) (3.68715084,0.158250601) (3.74301676,0.15389136) (3.79888268,0.1496522) (3.8547486,0.145529815) (3.91061453,0.141520986) (3.96648045,0.137622586) (4.02234637,0.133831573) (4.07821229,0.13014499) (4.13407821,0.126559958) (4.18994413,0.123073681) (4.24581006,0.119683439) (4.30167598,0.116386586) (4.3575419,0.11318055) (4.41340782,0.110062828) (4.46927374,0.107030989) (4.52513966,0.104082666) (4.58100559,0.101215559) (4.63687151,0.0984274299) (4.69273743,0.0957161044) (4.74860335,0.0930794661) (4.80446927,0.0905154579) (4.8603352,0.088022079) (4.91620112,0.0855973839) (4.97206704,0.0832394804) (5.02793296,0.0809465288) (5.08379888,0.0787167399) (5.1396648,0.0765483737) (5.19553073,0.0744397383) (5.25139665,0.0723891883) (5.30726257,0.0703951236) (5.36312849,0.0684559884) (5.41899441,0.0665702694) (5.47486034,0.0647364953) (5.53072626,0.0629532351) (5.58659218,0.0612190974) (5.6424581,0.059532729) (5.69832402,0.0578928141) (5.75418994,0.056298073) (5.81005587,0.0547472613) (5.86592179,0.053239169) (5.92178771,0.0517726193) (5.97765363,0.0503464678) (6.03351955,0.0489596017) (6.08938547,0.0476109388) (6.1452514,0.0462994268) (6.20111732,0.0450240423) (6.25698324,0.0437837901) (6.31284916,0.0425777024) (6.36871508,0.0414048381) (6.42458101,0.0402642821) (6.48044693,0.0391551444) (6.53631285,0.0380765594) (6.59217877,0.0370276856) (6.64804469,0.0360077046) (6.70391061,0.0350158204) (6.75977654,0.0340512591) (6.81564246,0.0331132681) (6.87150838,0.0322011153) (6.9273743,0.0313140891) (6.98324022,0.0304514974) (7.03910615,0.0296126669) (7.09497207,0.0287969433) (7.15083799,0.0280036899) (7.20670391,0.0272322879) (7.26256983,0.0264821352) (7.31843575,0.0257526466) (7.37430168,0.0250432528) (7.4301676,0.0243534002) (7.48603352,0.0236825507) (7.54189944,0.0230301807) (7.59776536,0.0223957811) (7.65363128,0.021778857) (7.70949721,0.021178927) (7.76536313,0.0205955228) (7.82122905,0.0200281895) (7.87709497,0.0194764841) (7.93296089,0.0189399762) (7.98882682,0.0184182472) (8.04469274,0.01791089) (8.10055866,0.0174175087) (8.15642458,0.0169377183) (8.2122905,0.0164711444) (8.26815642,0.016017423) (8.32402235,0.0155762) (8.37988827,0.0151471311) (8.43575419,0.0147298815) (8.49162011,0.0143241256) (8.54748603,0.0139295469) (8.60335196,0.0135458375) (8.65921788,0.0131726978) (8.7150838,0.0128098368) (8.77094972,0.0124569714) (8.82681564,0.0121138262) (8.88268156,0.0117801333) (8.93854749,0.0114556326) (8.99441341,0.0111400706) (9.05027933,0.0108332013) (9.10614525,0.0105347852) (9.16201117,0.0102445893) (9.21787709,0.0099623873) (9.27374302,0.00968795897) (9.32960894,0.00942109016) (9.38547486,0.00916157263) (9.44134078,0.0089092039) (9.4972067,0.00866378702) (9.55307263,0.0084251305) (9.60893855,0.00819304813) (9.66480447,0.00796735879) (9.72067039,0.0077478864) (9.77653631,0.00753445969) (9.83240223,0.00732691212) (9.88826816,0.00712508176) (9.94413408,0.0069288111) (10,0.006737947)};
\addlegendentry{Exact $e^{-\gamma |h|}$}
\addplot[only marks, mark=*, mark size=1.15pt, plotblue2, draw=plotblue2] coordinates {(0,1) (0.2,0.903715214) (0.4,0.815508421) (0.6,0.735443775) (0.8,0.664038203) (1,0.599607743) (1.2,0.542131861) (1.4,0.489852842) (1.6,0.443118976) (1.8,0.401338356) (2,0.364365899) (2.2,0.33183392) (2.4,0.303063562) (2.6,0.277461825) (2.8,0.255272214) (3,0.235003554) (3.2,0.217019398) (3.4,0.201338377) (3.6,0.185389638) (3.8,0.170939252) (4,0.157139687) (4.2,0.143816811) (4.4,0.130732872) (4.6,0.118517724) (4.8,0.106849647) (5,0.095869012) (5.2,0.0863897191) (5.4,0.0790411494) (5.6,0.0727464862) (5.8,0.0675546836) (6,0.0631935119) (6.2,0.0595536482) (6.4,0.0565164068) (6.6,0.0534627632) (6.8,0.0505407806) (7,0.0476330368) (7.2,0.0456507929) (7.4,0.0442346225) (7.6,0.0424254989) (7.8,0.0413317355) (8,0.0399321386) (8.2,0.0385378493) (8.4,0.0365583485) (8.6,0.0340875011) (8.8,0.0321444847) (9,0.0306092779) (9.2,0.0286835862) (9.4,0.0270101701) (9.6,0.0247334106) (9.8,0.0213642932)};
\addlegendentry{Estimated $\widehat{\rho}(h)$}
\addplot+[black!45, thin, no marks] coordinates {(0,0) (10,0)};
\end{axis}
\end{tikzpicture}
}
  \hfill
  \figurepanel{0.48\textwidth}{b}{Data-driven and direct Fano estimates}%
    {%
\begin{tikzpicture}
\begin{axis}[
  publication axis,
  width=8.5cm,
  height=5.0cm,
  xlabel={Fano factor $\Fbb[n]$},
  ylabel={Count (100 realizations)},
  xmin=1.07597495, xmax=1.19327247,
  ymin=0, ymax=72,
  legend style={
    at={(0.5,1.02)},
    anchor=south,
    legend columns=-1,
    draw=none,
    fill=none,
    /tikz/every even column/.append style={column sep=4pt}
  }
]
\addplot[forget plot, fill=plotred, fill opacity=0.48, draw=plotred!70!black, line width=0.35pt, mark=none, ybar interval] coordinates {(1.08497495,1) (1.09081598,1) (1.09665701,2) (1.10249805,2) (1.10833908,6) (1.11418011,10) (1.12002114,7) (1.12586217,11) (1.1317032,13) (1.13754423,11) (1.14338526,11) (1.14922629,12) (1.15506732,5) (1.16090835,3) (1.16674938,2) (1.17259041,1) (1.17843144,2) (1.18427247,0)};
\addlegendimage{area legend, fill=plotred, fill opacity=0.48, draw=plotred!70!black, line width=0.35pt}
\addlegendentry{Direct}
\addplot[forget plot, fill=plotblue2, fill opacity=0.48, draw=plotblue3, line width=0.35pt, mark=none, ybar interval] coordinates {(1.08497495,0) (1.09081598,0) (1.09665701,0) (1.10249805,0) (1.10833908,0) (1.11418011,0) (1.12002114,0) (1.12586217,26) (1.1317032,61) (1.13754423,13) (1.14338526,0) (1.14922629,0) (1.15506732,0) (1.16090835,0) (1.16674938,0) (1.17259041,0) (1.17843144,0) (1.18427247,0)};
\addlegendimage{area legend, fill=plotblue2, fill opacity=0.48, draw=plotblue3, line width=0.35pt}
\addlegendentry{Data-driven}
\addplot[black, densely dashed, very thick, no marks] coordinates {(1.13333333,0) (1.13333333,68.93)};
\addlegendentry{OU benchmark $\Fbb[n]=1.133$}
\end{axis}
\end{tikzpicture}
}
  \caption{\textbf{Validation of the data-driven recipe using a stationary Ornstein--Uhlenbeck benchmark.} The parameters are $\lambda_0=5$, $\sigma=1$, $\gamma=0.5$, and $\mu=1$, giving $\Ebb[n]=5$ and $\Fbb^{\mathrm{OU}}[n]=17/15$ from equation~\eqref{eq:ou-fano}. \textbf{(a)} Autocorrelation estimated from one sampled rate trajectory (blue points) and the exact OU result $\rho(h)=e^{-\gamma|h|}$ (black line). The lag $h$ is measured in units of the mRNA lifetime $1/\mu$; deviations at large $h$ occur after the correlation has become small. \textbf{(b)} Distributions of the data-driven estimate (blue) and the direct copy-number estimate (red) across $100$ independent realizations of duration $T=10^4/\mu$, sampled at intervals $\Delta t=0.1/\mu$. The dashed line marks the OU benchmark. The data-driven estimate has the smaller run-to-run variation.}
  \label{fig:data-driven-validation}
\end{figure*}

We next consider models in which the dynamics of $\lambda(t)$ are specified microscopically and the Fano factor can be derived analytically.

\section{Exact Analytic Results: Discrete Production with an \texorpdfstring{$\mathrm{M/M/1}$}{M/M/1} Queue}
\label{sec:mm1-queue}
We first consider exponential mRNA degradation with a transcription rate governed by an $\mathrm{M/M/1}$ process. An $\mathrm{M/M/1}$ queue is a birth--death process with a constant arrival rate and a constant total service rate whenever the queue is not empty~\cite{ross2014introduction}. Here each ``customer'' represents one increment $\delta$ of transcription rate rather than one mRNA molecule. Let $m(t)$ denote the upstream state at time $t$. For every real $t$, $m(t)$ is a nonnegative integer. For a fixed rate increment $\delta>0$, assume
\begin{equation}
  \lambda(t) \mathrel{:=} \delta\,m(t).
  \label{eq:mm1-rate-state}
\end{equation}
Thus $m(t)$ is the number of transcription-rate increments present at time $t$, distinct from the mRNA copy number $n(t)$. The state $m(t)$ is upstream since it determines $\lambda(t)=\delta m(t)$; this rate then drives the downstream mRNA production process whose copy number is $n(t)$. 
\begin{samepage}
\noindent For a state $(n,m)$, the coupled process has the transitions
\begin{subequations}
  \label{eq:mm1-joint-transitions}
  \begin{align}
    (n,m) &\longrightarrow (n+1,m)
      && \text{at rate } \delta m,
      \label{eq:mm1-mrna-production}\\
    (n,m) &\longrightarrow (n-1,m)
      && \text{at rate } \mu n,
      && n>0,
      \label{eq:mm1-mrna-degradation}\\
    (n,m) &\longrightarrow (n,m+1)
      && \text{at rate } m^{\uparrow},
      \label{eq:mm1-rate-increase}\\
    (n,m) &\longrightarrow (n,m-1)
      && \text{at rate } m^{\downarrow},
      && m>0.
      \label{eq:mm1-rate-decrease}
  \end{align}
\end{subequations}
\end{samepage}
Panels~(a) and (b) of Fig.~\ref{fig:mm1-model} summarize the construction. The upstream state $m(t)$ sets the instantaneous mRNA production rate, while $n(t)$ records the downstream mRNA copy number. Horizontal moves in the joint state space change $n$, and vertical moves change $m$.

The two vertical transitions constitute the $\mathrm{M/M/1}$ process: an arrival adds one rate increment, and a service completion removes one. The single-server assumption appears in equation~\eqref{eq:mm1-rate-decrease} as the constant total downward rate $m^{\downarrow}$ for every $m>0$. Assume $0<m^{\uparrow}<m^{\downarrow}$ and define $r\mathrel{:=} m^{\uparrow}/m^{\downarrow}\in(0,1)$. The condition $r<1$ is the $\mathrm{M/M/1}$ stability condition.  It gives a stationary distribution and prevents $\lambda(t)$ from diverging. Note that at stationarity, $m$ denotes a random variable with the stationary distribution of $m(t)$.

For comparison with mRNA degradation, we define the queue relaxation time $\vartheta_p$, the degradation timescale $\vartheta_d$, and their ratio $k$ by 
\begin{equation}
  \begin{aligned}
    \vartheta_p
    &\mathrel{:=}
    \frac{1}{m^{\downarrow}-m^{\uparrow}}, \\
    \vartheta_d
    &\mathrel{:=}
    \frac{1}{\mu}, \\
    k
    &\mathrel{:=}
    \frac{\vartheta_p}{\vartheta_d}
    =\frac{\mu}{m^{\downarrow}-m^{\uparrow}}.
  \end{aligned}
  \label{eq:mm1-timescales}
\end{equation}
\noindent The relaxation time $\vartheta_p$ measures how quickly a perturbation of the queue state returns toward stationarity. 

\noindent The following result accounts for the complete $\mathrm{M/M/1}$ autocorrelation rather than approximating it by a single exponential. The derivation is given in Appendix~\ref{app:framework-proofs}.

\begin{proposition}[Fano factor for $\mathrm{M/M/1}$ production]
For exponential mRNA degradation and the stationary $\mathrm{M/M/1}$ transcription-rate process defined above, the exact Fano factor is
\begin{equation}
  \Fbb[n]
  = 1 + \Ebb[n]
  \left[
    \frac{1}{r}
    -\frac{(1-r)^2z}{r(1-rz)^2}
  \right],
  \label{eq:mm1-fano}
\end{equation}
where
\begin{equation}
  z
  \mathrel{:=}
  \frac{2}{a+\sqrt{a^2-4r}},
  \qquad
  a
  \mathrel{:=}
  1+r+k(1-r).
  \label{eq:mm1-auxiliaries}
\end{equation}
\end{proposition}

The term following $1$ in equation~\eqref{eq:mm1-fano} is the additional copy-number variation transmitted from fluctuations in $\lambda(t)$. For fixed $r$ and $\Ebb[n]$, the Fano factor increases monotonically with $k$ and satisfies
\begin{subequations}
  \label{eq:mm1-fano-bounds}
  \begin{align}
    \lim_{k\to0}\Fbb[n]
    &=1,
    \label{eq:mm1-fast-limit}\\
    \lim_{k\to\infty}\Fbb[n]
    &=1+\frac{\Ebb[n]}{r}.
    \label{eq:mm1-slow-limit}
  \end{align}
\end{subequations}
\noindent When $k\ll1$, the queue fluctuates rapidly relative to mRNA degradation and its contribution is averaged out. When $k\gg1$, the transcription rate is effectively constant over an mRNA lifetime, giving the slow-fluctuation ceiling.

The stationary mean copy number is
\begin{align}
  \Ebb[n] = \frac{\Ebb[\lambda]}{\mu} = \frac{r}{1-r}\,\frac{\delta}{\mu}\,.
  \label{eq:mm1-mean-copy-number}
\end{align}
At fixed $r$ and $\mu$, varying $\delta$ changes $\Ebb[n]$ without changing $k$, whereas scaling both queue-transition rates by the same factor changes $k$ without changing $r$. For fixed $r$ and $k$, the Fano factor in equation~\eqref{eq:mm1-fano} depends linearly on $\Ebb[n]$, with intercept $1$.  Figure~\ref{fig:mm1-model}(c) compares this prediction with direct simulation.

\begin{figure*}[!tbp]
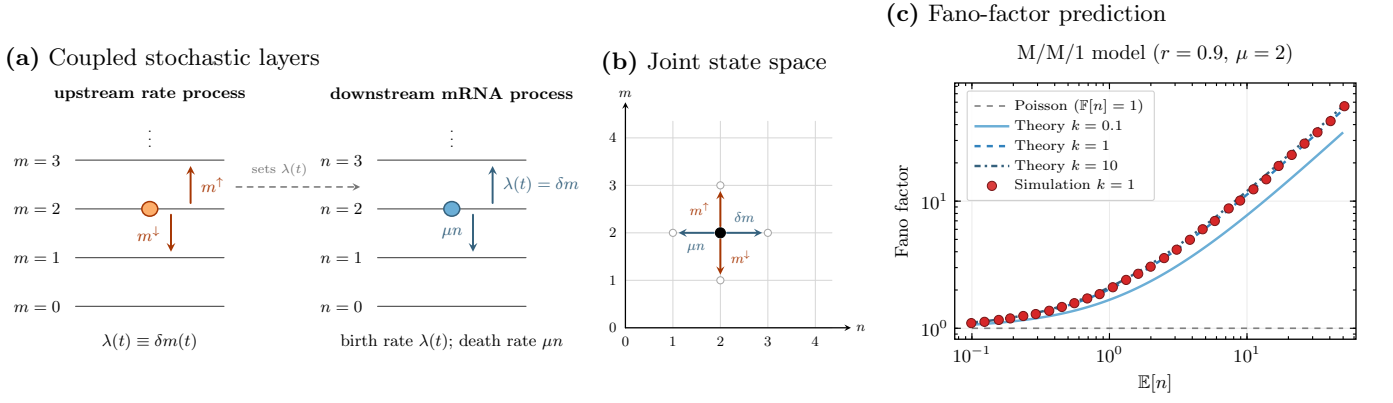

  \centering
  \figurepanel{0.43\textwidth}{a}{Coupled stochastic layers}%
    {%

}
  \hfill
  \figurepanel{0.20\textwidth}{b}{Joint state space}%
    {%
%
}
  \caption{\textbf{Discrete $\mathrm{M/M/1}$ transcription-rate model and its Fano factor.} \textbf{(a)} The upstream state $m(t)$ changes by one at rates $m^{\uparrow}$ and $m^{\downarrow}$ and sets $\lambda(t)=\delta m(t)$. Conditional on this rate, the downstream copy number increases at rate $\lambda(t)$ and decreases at total rate $\mu n$. \textbf{(b)} The same four transitions from a state $(n,m)$. The leftward transition is absent at $n=0$, and the downward transition is absent at $m=0$. \textbf{(c)} Equation~\eqref{eq:mm1-fano} for $k=0.1$, $1$, and $10$, at fixed $r=m^{\uparrow}/m^{\downarrow}=0.9$ and $\mu=2$. The grey dashed line marks $\Fbb[n]=1$. Red points are simulations at $k=1$, obtained by varying $\delta$, and agree with the exact curve.}
  \label{fig:mm1-model}
\end{figure*}

We implement a continuous-time simulation of the joint Markov process for the pair $(n,m)$ on the discrete state space introduced above. The four elementary transitions are mRNA production ($n\!\to\!n+1$, rate $\lambda\mathrel{:=}\delta\,m$), mRNA degradation ($n\!\to\!n-1$, rate $\mu n$), and increments and decrements of the production rate ($m\!\to\!m\pm1$, rates $m^{\uparrow}$ and $m^{\downarrow}$).  At each step, we draw a random waiting time for each transition from its exponential distribution, carry out whichever one would occur soonest, and advance the clock to that time; we run $2\times10^{7}$ such steps per simulation. The Fano factor is computed from the mean and variance of $n$ along the trajectory, with each inter-event value weighted by the duration for which that value persists. Fixing $\mu=2$, $r=0.9$, and $k=1$ (so that $m^{\downarrow}=20$ and $m^{\uparrow}=18$), we sweep the increment size $\delta$ to vary $\Ebb[n]$ over $[0.1,50]$ and find excellent agreement between the simulated Fano factor and the $k=1$ exact curve [equation~\eqref{eq:mm1-fano}]. At small mean copy number the extrinsic term is small and the Fano factor approaches the Poisson value $\Fbb[n]\simeq1$; at larger mean copy number the extrinsic contribution grows linearly. Details regarding the simulation are provided in Appendix~\ref{app:simulation-details}.

\section{Exact Analytic Results: Continuous Production with Drift--Diffusion}
\label{sec:drift-diffusion}
Next, we consider Poisson degradation with a continuous transcription rate $\lambda(t)$. Let $p(\lambda,t\mid\lambda_0,t_0)$ denote the conditional probability density of $\lambda(t)$ given $\lambda(t_0)=\lambda_0$. We model its evolution by the drift--diffusion equation
\begin{equation}
  \begin{aligned}
    \label{eq:smoluchowski}
    \frac{\partial}{\partial t}p(\lambda,t\mid\lambda_0,t_0)
    &=D\frac{\partial^2}{\partial\lambda^2}p(\lambda,t\mid\lambda_0,t_0) \\
    &+v\frac{\partial}{\partial\lambda}p(\lambda,t\mid\lambda_0,t_0),
  \end{aligned}
\end{equation}
where $D>0$ is the diffusion constant and $v$ is a signed drift parameter. Equation~\eqref{eq:smoluchowski} is equivalent to
\begin{equation}
  d\lambda=-v\,dt+\sqrt{2D}\,dW.
  \label{eq:drift-diffusion-langevin}
\end{equation}
Equation~\eqref{eq:smoluchowski} describes the evolution of the probability density, whereas equation~\eqref{eq:drift-diffusion-langevin} describes individual stochastic trajectories. They are equivalent descriptions of the same drift--diffusion process; we use the density equation analytically and the Langevin equation in simulations. Thus the physical drift velocity is $-v$: positive $v$ drives $\lambda(t)$ downward, whereas negative $v$ drives it upward. The permitted values of $\lambda(t)$ and its behavior at the endpoints are determined by the boundary conditions considered below. Figure~\ref{fig:reflecting-model}(a) provides a schematic of the drift--diffusion process and boundary conditions. 

\subsection{Reflecting Boundary Condition}
We first take $v>0$, so the physical drift $-v$ points toward zero. We restrict the transcription rate to nonnegative values by imposing a reflecting boundary at $\lambda=0$.  Suppressing the arguments of $p$, the corresponding zero-flux condition is
\begin{equation}
  \left.(D\,\partial_\lambda+v)p\right|_{\lambda=0}=0.
  \label{eq:reflecting-bc}
\end{equation}
Figure~\ref{fig:reflecting-model}(a) summarizes the drift convention and the reflecting boundary.

For the elapsed time $\Delta t\mathrel{:=} t-t_0>0$, the conditional probability density is~\cite{risken1989fokker,redner2001guide}
\begin{align}
  p
  &=\frac{1}{\sqrt{4\pi D\Delta t}}\left[
    e^{-\frac{(\lambda-\lambda_0+v\Delta t)^2}{4D\Delta t}}
    +e^{\frac{v\lambda_0}{D}-\frac{(\lambda+\lambda_0+v\Delta t)^2}{4D\Delta t}}
  \right] \notag\\
  &+\frac{v}{2D}e^{-\frac{v\lambda}{D}}
    \,\mathrm{erfc}\!\left[
      \frac{\lambda+\lambda_0-v\Delta t}{\sqrt{4D\Delta t}}
    \right].
  \label{eq:reflecting-propagator}
\end{align}
The stationary production rate has density $p_s(\lambda)=(v/D)e^{-v\lambda/D}$ for $\lambda\geq0$.  Therefore $\Ebb[\lambda]=D/v$, $\Vbb[\lambda]=D^2/v^2$, and
$\Ebb[n]=\Ebb[\lambda]/\mu=D/(\mu v)$. Let us define the characteristic production timescale $\vartheta_p\mathrel{:=} D/v^2$, the degradation timescale $\vartheta_d\mathrel{:=}1/\mu$, and their ratio $k\mathrel{:=}\vartheta_p/\vartheta_d=\mu D/v^2$. Substitution into equation~\eqref{eq:poisson-fano} gives
\begin{align}
  \label{eq:reflecting-fano}
  \Fbb[n] = 1 + \Ebb[n] \left[\frac{\sqrt{1 + 4 k}-1}{2k^2} - \frac{1}{k} + 1\right].
\end{align}
A derivation is given in Appendix~\ref{app:reflecting-derivation}.

The bracketed factor in equation~\eqref{eq:reflecting-fano} increases monotonically with $k$. When $k\to\infty$, the production rate varies slowly relative to mRNA degradation and is effectively constant over an mRNA lifetime. Since the stationary rate distribution satisfies $\Vbb[\lambda]/\Ebb[\lambda]^2=1$, the Fano factor approaches
\begin{align}
  \Fbb[n]\big|_{k \rightarrow \infty} = 1 + \Ebb[n].
\end{align}
When $k\to0$, the production rate fluctuates rapidly relative to mRNA degradation and its contribution is averaged out, giving
\begin{align}
  \Fbb[n]\big|_{k \rightarrow 0} = 1.
\end{align}
Thus equation~\eqref{eq:reflecting-fano} interpolates between the Poisson floor and the slow-fluctuation ceiling $1+\Ebb[n]$. 

\begin{figure*}[!tbp]
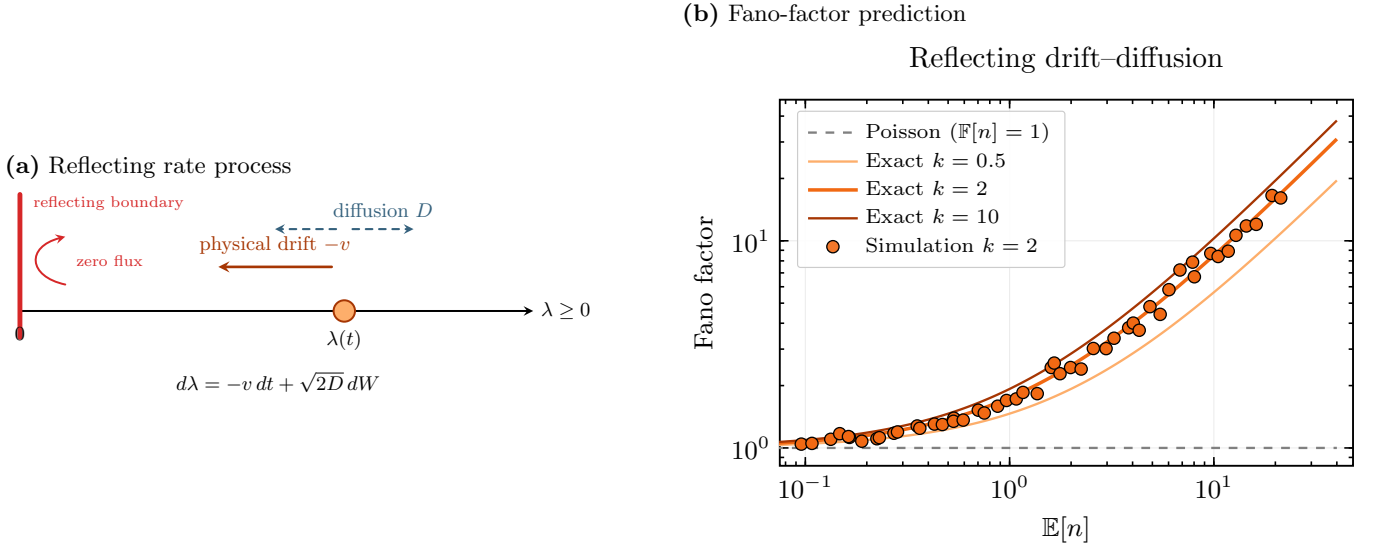

  \centering
  \figurepanel{0.44\textwidth}{a}{Reflecting rate process}%
    {%

}
  \caption{\textbf{Reflecting drift--diffusion transcription-rate model and its Fano factor.} \textbf{(a)} Diffusion with constant $D$ produces fluctuations in $\lambda(t)$, while the physical drift $-v$ points toward the origin for $v>0$. The zero-flux condition at $\lambda=0$ confines the process to nonnegative rates. \textbf{(b)} Solid curves show equation~\eqref{eq:reflecting-fano} for $k=0.5$, $2$, and $10$; the grey dashed line marks $\Fbb[n]=1$. As $k$ increases, the curves approach the slow-fluctuation ceiling $\Fbb[n]=1+\Ebb[n]$. Orange points are simulations at $k=2$ and $\mu=2$, obtained by varying $D$ and $v$ at fixed $k$.}
  \label{fig:reflecting-model}
\end{figure*}

To test equation~\eqref{eq:reflecting-fano}, we simulate the Langevin process~\eqref{eq:drift-diffusion-langevin} with a reflecting boundary at $\lambda=0$, implemented after each time step by replacing a proposed negative value with its absolute value.  The rate process is simulated jointly with the mRNA birth--death process. We fix $\mu=2$ and $k=2$. For each target mean, we set $v=\mu^2\Ebb[n]/k$ and $D=\mu^3\Ebb[n]^2/k$; this varies $\Ebb[n]=D/(\mu v)$ while keeping $\vartheta_p=D/v^2=k/\mu$ fixed. The resulting Fano factors agree with equation~\eqref{eq:reflecting-fano}, as shown in Fig.~\ref{fig:reflecting-model}(b). Simulation details are given in Appendix~\ref{app:simulation-details}.

The reflecting drift--diffusion process is the continuum limit of the $\mathrm{M/M/1}$ rate model. As $\delta\to0$, we scale the upward and downward transition rates so that $v\mathrel{:=}(m^{\downarrow}-m^{\uparrow})\delta$ and $D\mathrel{:=}\tfrac12(m^{\uparrow}+m^{\downarrow})\delta^2$ remain fixed. The discrete process then converges to reflected Brownian motion with physical drift $-v$ and diffusion constant $D$. Because $r<1$ means $m^{\uparrow}<m^{\downarrow}$ and $\delta>0$, the scaling relation $v=(m^{\downarrow}-m^{\uparrow})\delta$ gives $v>0$.  Note that the discrete and continuous Fano factors have the same limits as $k\to0$ and $k\to\infty$, but differ at intermediate $k$ because the $\mathrm{M/M/1}$ production rate changes in finite increments.

\subsection{Periodic Boundary Conditions}
\label{sec:periodic-boundary}
During the cell cycle, cellular state and gene dosage change systematically with cell age, and division begins a new cycle. This motivates transcription-rate dynamics that repeatedly cycle through the same range~\cite{padovan2015single,golding2024primer}. We first consider a tractable periodic model: confine $\lambda(t)$ to $[0,L]$ and impose periodic wrapping. A trajectory that crosses $L$ re-enters at $0$, while one that crosses $0$ re-enters at $L$; the interval therefore behaves like a ring. This periodic model is a tractable mathematical benchmark rather than a literal model of cell division. It permits wrapping in both directions, whereas biological division is directional. The one-way reset model is introduced in Subsection~\ref{sec:first-passage-resetting}. The periodic boundary conditions are
\begin{align}
  p(0,t\,|\,\lambda_0,t_0)
  &=p(L,t\,|\,\lambda_0,t_0),\notag \\
  \partial_\lambda p(0,t\,|\,\lambda_0,t_0)
  &=\partial_\lambda p(L,t\,|\,\lambda_0,t_0).
  \label{eq:periodic-bc}
\end{align}
The stationary density is uniform, $p_s(\lambda)=1/L$. Therefore $\Ebb[\lambda]=L/2$ and $\Vbb[\lambda]=L^2/12$. The stationary covariance
$C(h)\mathrel{:=}\Ebb[\lambda(t+h)\lambda(t)]-\Ebb[\lambda]^2$ is
\begin{align}
  C(h)=\frac{L^2}{2\pi^2}
  \sum_{\ell=1}^{\infty}\frac{1}{\ell^2}\,
  e^{-4\pi^2 \ell^2 D |h|/L^2}\cos\!\left(\frac{2\pi \ell v h}{L}\right),
  \label{eq:periodic-autocorr}
\end{align}
as derived in Appendix~\ref{app:periodic-derivation}. Let us define the slowest diffusive relaxation time $\vartheta_p$ and the circulation frequency $\omega$ by
\begin{align}
  \vartheta_p&\mathrel{:=}\frac{L^2}{4\pi^2D},
  \qquad
  \omega\mathrel{:=}\frac{2\pi v}{L}.
  \label{eq:periodic-scales}
\end{align}
For $h\geq0$, mode $\ell$ in equation~\eqref{eq:periodic-autocorr} decays as $e^{-\ell^2h/\vartheta_p}$ and oscillates as $\cos(\ell\omega h)$. The mode $\ell=1$ is the slowest, and the dimensionless product $\vartheta_p\omega=vL/(2\pi D)$ measures circulation relative to diffusive relaxation.

For exponential degradation, define $\vartheta_d\mathrel{:=}1/\mu$ and $k\mathrel{:=}\vartheta_p/\vartheta_d=\mu\vartheta_p$. The stationary mean copy number is $\Ebb[n]=L/(2\mu)$. Substituting equation~\eqref{eq:periodic-autocorr} into equation~\eqref{eq:poisson-fano} gives the exact Fano factor
\begin{align}
  \Fbb[n]
  &=1 \notag\\
  &+\Ebb[n]\,\frac{2k}{\pi^2}
  \sum_{\ell=1}^{\infty}
  \frac{k+\ell^2}{\ell^2\left[(k+\ell^2)^2+\ell^2\vartheta_p^2\omega^2\right]}.
  \label{eq:periodic-fano}
\end{align}
For a simpler approximation, we replace the full covariance by $(L^2/12)e^{-|h|/\vartheta_p}\cos(\omega h)$. This approximation retains the decay rate and frequency of the $\ell=1$ Fourier mode, omits the faster modes, and assigns the retained mode the complete stationary variance $L^2/12$, so that it preserves the exact value $C(0)=\Vbb[\lambda]$.  Equation~\eqref{eq:poisson-fano} then gives
\begin{align}
  \Fbb[n]
  &\approx1+\frac{\Ebb[n]}{3}
  \frac{k(k+1)}{(k+1)^2+\vartheta_p^2\omega^2}.
  \label{eq:periodic-fano-single}
\end{align}

When $k\to\infty$, the production rate changes slowly relative to mRNA degradation and the Fano factor approaches the slow-fluctuation ceiling $1+\Ebb[n]/3$. At stationarity, $\lambda$ is uniformly distributed on $[0,L]$, and therefore
\begin{equation*}
  \frac{\Vbb[\lambda]}{\Ebb[\lambda]^2}
  =\frac{L^2/12}{(L/2)^2}
  =\frac{1}{3}.
\end{equation*}
\noindent When $k\to0$, rapid rate fluctuations are averaged out and $\Fbb[n]\to1$. At fixed $k$, increasing $|\vartheta_p\omega|$ also suppresses the transmitted rate fluctuations because the oscillatory parts of $C(h)$ cancel in the time integral.

\begin{figure*}[!tbp]
  \centering
  \figurepanel{0.48\textwidth}{a}{No drift: $\vartheta_p\omega=0$}%
    {%

}
  \caption{\textbf{Fano factor for periodic drift--diffusion production.} Solid curves show the exact result in equation~\eqref{eq:periodic-fano}; the grey dashed line marks the Poisson value $\Fbb[n]=1$. \textbf{(a)} Without drift, increasing $k$ moves the curves toward the slow-fluctuation ceiling $1+\Ebb[n]/3$ (dark dashed). \textbf{(b)} For $\vartheta_p\omega=0.6\Ebb[n]$, circulation becomes faster as the mean increases. It eventually averages out the rate fluctuations, causing each curve to turn downward. Green points show stochastic simulations at $k=10$ and $\mu=1$. They agree with the corresponding exact curves in both panels.}
  \label{fig:periodic-fano}
\end{figure*}

To test equation~\eqref{eq:periodic-fano}, we simulate equation~\eqref{eq:drift-diffusion-langevin} with periodic wrapping, jointly with the mRNA birth--death process. We fix $\mu=1$ and $k=10$, vary $L$, and set $D=\mu L^2/(4\pi^2k)$. Panel (a) uses $v=0$. Panel (b) uses $v=0.6\pi D/\mu$, which gives $\vartheta_p\omega=0.6\Ebb[n]$. Simulation details are given in Appendix~\ref{app:simulation-details}.

\subsection{First-Passage Resetting}
\label{sec:first-passage-resetting}
We now replace periodic wrapping by a one-way reset. Let $L\mathrel{:=}\lambda_{\max}-\lambda_{\min}$. Between resets, $\lambda(t)$ follows equation~\eqref{eq:drift-diffusion-langevin} on $[\lambda_{\min},\lambda_{\max}]$ and reflects at $\lambda_{\min}$. When it first reaches $\lambda_{\max}$, it is reset to $\lambda_{\min}$ and a new cycle begins. This model describes a transcription rate that builds up and then drops abruptly, as may occur across cell division~\cite{padovan2015single}. The physical drift is upward when $v<0$, because the drift velocity in equation~\eqref{eq:drift-diffusion-langevin} is $-v$. Before the upper boundary is reached, the single-cycle propagator obeys the reflecting and absorbing boundary conditions
\begin{align}
  D\partial_\lambda p+vp&=0
  \;\text{at }\lambda=\lambda_{\min},
  \quad
  p(\lambda_{\max},t)&=0,
  \label{eq:fpr-bc}
\end{align}
The full process is obtained by reinjecting the probability absorbed at $\lambda_{\max}$ at $\lambda_{\min}$.

Define the dimensionless position $\xi\mathrel{:=}(\lambda-\lambda_{\min})/L$, the dimensionless drift $\alpha\mathrel{:=} vL/D$, and $\psi(\alpha)\mathrel{:=}(e^\alpha-1)/\alpha-1$. A stationary flux-balance calculation gives
\begin{align}
  p_s(\lambda)
  =\frac{e^{\alpha(1-\xi)}-1}{L\,\psi(\alpha)},
  \label{eq:fpr-steady}
\end{align}
as derived in Appendix~\ref{app:fpr-derivation}. For pure diffusion, $\alpha\to0$ and $p_s(\lambda)\to2(1-\xi)/L$. Under strong upward drift, $\alpha\to-\infty$ and $p_s(\lambda)\to1/L$. The first two moments of $\xi$ are
\begin{align}
  \Ebb[\xi]&=\frac{1}{\alpha}-\frac{1}{2\psi(\alpha)},\\
  \Vbb[\xi]&=\frac{1}{\alpha^2}
  -\frac{1}{3\psi(\alpha)}
  -\frac{1}{4\psi(\alpha)^2},
  \label{eq:fpr-moments}
\end{align}
and therefore
\begin{align*}
  \Ebb[\lambda]&=\lambda_{\min}+L\Ebb[\xi],\\
  \Vbb[\lambda]&=L^2\Vbb[\xi],\\
  \Ebb[n]&=\frac{\Ebb[\lambda]}{\mu}.
\end{align*}
Let $T_{\mathrm{cell}}$ denote the random time between consecutive resets. Its mean is
\begin{align}
  \Ebb[T_{\mathrm{cell}}]
  =\frac{L^2}{\alpha D}\,\psi(\alpha).
  \label{eq:fpr-mfpt}
\end{align}
For pure diffusion this becomes $L^2/(2D)$; under strong upward drift it approaches $L/|v|$.

A simple approximation replaces the stationary autocorrelation by one exponential. For $\alpha<2$, let $k_0$ be the smallest positive solution of $k_0\cot(k_0L)=v/(2D)$ and define $E_0\mathrel{:=} Dk_0^2+v^2/(4D)$, $\vartheta_\lambda\mathrel{:=}1/E_0$, and $k_\lambda\mathrel{:=}\mu\vartheta_\lambda$. Approximating $\rho(h)$ by $e^{-|h|/\vartheta_\lambda}$ in equation~\eqref{eq:poisson-fano} gives
\begin{align}
  \Fbb[n]
  \approx1+\Ebb[n]
  \frac{\Vbb[\lambda]}{\Ebb[\lambda]^2}
  \frac{k_\lambda}{k_\lambda+1}.
  \label{eq:fpr-fano}
\end{align}
This approximation neglects the faster modes and correlations across reset events. It is included for comparison in Fig.~\ref{fig:fpr-fano}, but the solid curves use the exact result below.

For the exact calculation, let us define $C(h)\mathrel{:=}\Ebb[\lambda(t+h)\lambda(t)]-\Ebb[\lambda]^2$ and $J(\mu)$ as its exponentially weighted integral. Equation~\eqref{eq:poisson-fano} then becomes
\begin{align}
  \Fbb[n]
  =1+\frac{J(\mu)}{\Ebb[\lambda]},
  \qquad
  J(\mu)\mathrel{:=}\int_0^\infty C(h)\,e^{-\mu h}\,d h.
  \label{eq:fpr-fano-j}
\end{align}
The integral can be evaluated without first deriving $C(h)$. 
Let $g(\lambda):=\lambda-\Ebb[\lambda]$. For an initial rate $\lambda$, we define
\begin{equation*}
  \chi(\lambda)
  :=
  \int_0^\infty e^{-\mu h}
  \Ebb\!\left[g(\lambda(h))\mid\lambda(0)=\lambda\right]\,d h.
\end{equation*}
Thus $\chi(\lambda)$ is the expected future deviation of the transcription rate from its stationary mean, with the deviation at lag $h$ weighted by the mRNA survival probability $e^{-\mu h}$. Equivalently, $\chi=(\mu-\mathcal{L})^{-1}g$ and $J(\mu)=\Ebb_{p_s}[g\chi]$. For the present drift--diffusion process, the corresponding boundary-value problem is
\begin{equation}
\begin{aligned}
  D\chi''-v\chi'-\mu\chi
  &=-(\lambda-\Ebb[\lambda]),\\
  \chi'(\lambda_{\min})&=0,\\
  \chi(\lambda_{\max})&=\chi(\lambda_{\min}).
\end{aligned}
  \label{eq:fpr-bvp}
\end{equation}
Averaging $\chi(\lambda)$ over the stationary density, weighted by the current deviation $g(\lambda)$, gives 
\begin{align*}
  J(\mu)
  =\int_{\lambda_{\min}}^{\lambda_{\max}}
  (\lambda-\Ebb[\lambda])\chi(\lambda)p_s(\lambda)\,d\lambda.
\end{align*}
\noindent
This boundary-value problem includes all relaxation modes and correlations across resets; its derivation is given in Appendix~\ref{app:resolvent}.

For pure diffusion with $\lambda_{\min}=0$, let us define $u\mathrel{:=} L\sqrt{\mu/D}$. Solving equation~\eqref{eq:fpr-bvp} gives
\begin{align}
  \Fbb[n]
  =1+\sqrt{\frac{D}{\mu^3}}
  &\left[
    \frac{u}{6}-\frac{6}{u}\coth^2\!\left(\frac{u}{2}\right)+
  \right.
  \notag\\
  &\left.
    \;\frac{24}{u^2}\coth\!\left(\frac{u}{2}\right)
    -\frac{24}{u^3}
  \right].
  \label{eq:fpr-closed-a0}
\end{align}
As $u\to0$, the sum of the two correction terms vanishes and $\Fbb[n]\to1$. As $u\to\infty$, its leading contribution is $\sqrt{D/\mu^3}\,u/6$. Since $\Ebb[n]=L/(3\mu)=\sqrt{D/\mu^3}\,u/3$ in this case, the Fano factor approaches the slow-fluctuation ceiling $1+\Ebb[n]/2$.

\begin{figure}[!tbp]
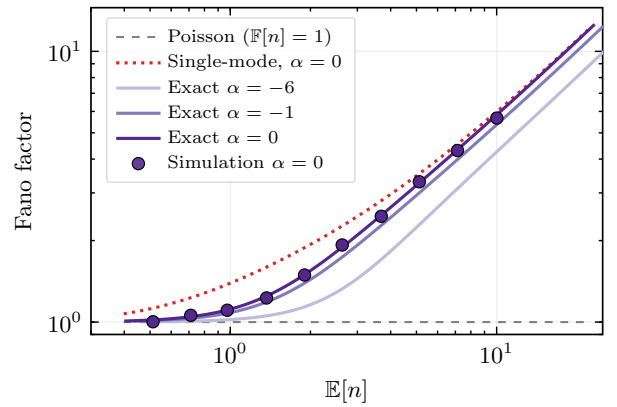

  \centering
  \tikzfigure{8cm}{%

}
  \caption{\textbf{Fano factor for drift--diffusion production with first-passage resetting.} Solid curves show the exact result from equations~\eqref{eq:fpr-fano-j} and~\eqref{eq:fpr-bvp} for $\alpha=-6$, $-1$, and $0$; the grey dashed line marks the Poisson value $\Fbb[n]=1$. The red dotted curve is the single-mode approximation~\eqref{eq:fpr-fano} for pure diffusion, $\alpha=0$. More negative $\alpha$ produces stronger upward drift and a smaller slow-fluctuation slope. Purple points show stochastic simulations at $\alpha=0$ and agree with the corresponding exact curve.}
  \label{fig:fpr-fano}
\end{figure}

Figure~\ref{fig:fpr-fano} uses $D=\mu=1$ and $\lambda_{\min}=0$. The solid curves are obtained by varying $L$ at fixed $\alpha\in\{-6,-1,0\}$. The simulation points use $\alpha=0$ and ten logarithmically spaced values of $L$ from $1.5$ to $30$. Simulation details are given in Appendix~\ref{app:simulation-details}.

\section{Conclusions}
\label{sec:conclusion}
We derived a general relation between fluctuations in the transcription rate $\lambda(t)$ and fluctuations in the mRNA copy number $n(t)$. Under mRNA production that is Poisson conditional on the complete transcription-rate trajectory $\lambda(\cdot)$, independent mRNA lifetimes, and a stationary transcription-rate process, equation~\eqref{eq:general-fano} expresses the Fano factor in terms of the lifetime survival probability, $\Ebb[\lambda]$, $\Vbb[\lambda]$, and $\rho(h)$.  For exponential degradation, equation~\eqref{eq:poisson-fano} reduces this result to an exponentially weighted integral of $\rho(h)$. The additive term $1$ is the Poisson copy-number variation that remains when the transcription rate is constant. The remaining term is the variation transmitted from fluctuations in $\lambda(t)$.

The result can be used in two ways. First, it can be applied to a specified model for $\lambda(t)$. We recovered the constitutive, random-static, telegraph, and cell-cycle-null cases in Section~\ref{sec:existing-models}. We then derived results for the discrete $\mathrm{M/M/1}$ process in Section~\ref{sec:mm1-queue} and for drift--diffusion with reflecting, periodic, and first-passage-reset boundaries in Section~\ref{sec:drift-diffusion}. In each model, the relative timescales of transcription-rate fluctuations and mRNA degradation determine how much rate variation reaches the copy number. Rapid rate fluctuations are averaged out, and the Fano factor approaches the Poisson value $1$. Slowly varying rates approach the slow-fluctuation ceiling. For first-passage resetting, the generator resolvent provides the exact result when a single-exponential approximation to the autocorrelation is insufficient. Starting from a given rate, this resolvent computes the expected future deviation from the stationary mean, discounted over the mRNA degradation timescale. Averaging that response over the stationary distribution gives the rate variation transmitted to the copy number. 

Second, equation~\eqref{eq:poisson-fano} can be used without specifying a parametric model for $\lambda(t)$.  If a measured rate trajectory is stationary and ergodic, its mean, variance, and autocorrelation can be estimated directly and substituted into equation~\eqref{eq:poisson-fano}. The synthetic example in Section~\ref{sec:data-driven-validation} recovered the exact Fano factor and had less sampling variation than the direct estimate obtained from the mRNA copy-number trajectory. This approach requires a known degradation rate $\mu$, but it does not require a parametric model or a closed-form expression for the autocorrelation.

The general result also permits non-exponential mRNA lifetimes by replacing the exponential survival probability with the appropriate function $\overline{\varphi}(s)$, provided the required integrals exist. The present analysis is limited to the mean and variance and assumes production that is Poisson conditional on $\lambda(\cdot)$ and independent lifetimes.  Extending it to higher moments or the full copy-number distribution would distinguish models that have the same $\Ebb[n]$ and $\Fbb[n]$ but different distributional shapes. A direct experimental test would estimate $\lambda(t)$ from live-cell measurements, predict the mRNA Fano factor, and compare it with matched copy-number data.

\begin{acknowledgments}
  This work was funded in part by the National Institutes of Health under Award No.~1R35GM162492-01.
  Work in the Golding lab is supported by the National Institutes of Health grant R35 GM140709, the National Science Foundation grant 2243257 (NSF Science and Technology Center for Quantitative Cell Biology), and the Alfred P. Sloan Foundation grant G-2023-19649.
\end{acknowledgments}

\appendix

\section{Derivations of the Doubly Stochastic Results}
\label{app:framework-proofs}

In brief, this appendix shows the following: Poisson thinning gives the conditional copy-number distribution; stationarity and the law of total variance give the mean and general Fano factor; exponential lifetimes reduce the double integral to a single weighted autocorrelation integral; and the Cauchy--Schwarz inequality gives the Poisson floor and slow-fluctuation ceiling. The final derivation in this appendix evaluates that weighted autocorrelation exactly for the $\mathrm{M/M/1}$ rate process.

\begin{proof}[Derivation of equations~\eqref{eq:newell} and~\eqref{eq:layered-conditional-poisson}]
\mbox{}\par\noindent
Consider the mRNAs produced approximately $s$ time units before the observation time $t$, within an interval of width $d s$. Conditional on $\lambda(\cdot)$, their expected number is $\lambda(t-s)\,d s$, and each remains present at time $t$ with probability $\overline{\varphi}(s)$. This group therefore contributes an expected $\overline{\varphi}(s)\lambda(t-s)\,d s$ molecules. Integrating over all molecular ages gives the conditional mean. For a Poisson number of produced mRNAs, independently retaining each molecule with its survival probability leaves a Poisson number of survivors. This property is known as Poisson thinning.  Moreover, production counts from nonoverlapping time intervals are independent, and independently thinning each interval preserves the independence of the resulting survivor counts. Their total is consequently Poisson with the stated mean.
\end{proof}

\begin{proof}[Derivation of equation~\eqref{eq:layered-stationary-mean}]
Taking the expectation of the conditional mean and moving the expectation through the integral gives
\begin{align}
  \Ebb[n]
  &= \int_0^\infty \overline{\varphi}(s)\,
  \Ebb[\lambda(t-s)]\,d s \notag\\
  &= \Ebb[\lambda]
  \int_0^\infty\overline{\varphi}(s)\,d s \notag\\
  &= \Ebb[\lambda]\Ebb[\tau]\,.
\end{align}
The second line follows from stationarity. For each realized value of the nonnegative lifetime $\tau$, its value is the length of the interval from $0$ to $\tau$:
\begin{equation*}
  \tau=\int_0^\tau d s.
\end{equation*}
Averaging over the lifetime distribution and exchanging the order of averaging and integration gives
\begin{equation*}
  \Ebb[\tau]
  =\int_0^\infty \Pbb(\tau>s)\,d s
  =\int_0^\infty\overline{\varphi}(s)\,d s.
\end{equation*}
\end{proof}

\begin{proof}[Derivation of equation~\eqref{eq:general-fano}]
The conditional Poisson distribution gives
\begin{equation}
  \Vbb\!\left[n\,\middle|\,\lambda(\cdot)\right]
  =
  \Ebb\!\left[n\,\middle|\,\lambda(\cdot)\right]\,.
  \label{eq:layered-conditional-poisson-variance}
\end{equation}
The law of total variance then gives
\begin{equation}
  \Vbb[n]
  = \Ebb[n]
  +\Vbb\!\left[
    \int_0^\infty \overline{\varphi}(s)\,
    \lambda(t-s)\,d s
  \right]\,.
  \label{eq:layered-total-variance}
\end{equation}
Expanding the remaining variance and using stationarity yields
\begin{align}
  &\Vbb\!\left[
    \int_0^\infty
    \overline{\varphi}(s)\,
    \lambda(t-s)\,d s
  \right] \notag\\
  &={}\int_0^\infty\!\!\int_0^\infty
  \overline{\varphi}(s_1)\overline{\varphi}(s_2)
  \Bigl[\Ebb[\lambda(s_1)\lambda(s_2)]
  -\Ebb[\lambda]^2\Bigr]
  \,d s_1\,d s_2 \notag\\
  &={}\Vbb[\lambda]\int_0^\infty\!\!\int_0^\infty
  \overline{\varphi}(s_1)\overline{\varphi}(s_2)
  \rho(|s_2-s_1|)
  \,d s_1\,d s_2\,.
  \label{eq:layered-conditional-mean-variance}
\end{align}
The normalized autocorrelation depends only on the time difference and is even for a real stationary process. Substitution into the total-variance expression and division by $\Ebb[n]$ give equation~\eqref{eq:general-fano}.
\end{proof}

\begin{proof}[Derivation of equation~\eqref{eq:poisson-fano}]
The symmetry of the general double integral allows its domain to be split along $s_1=s_2$. On either half, writing the difference between the two molecular ages as $h$ gives
\begin{align}
  &\int_0^\infty\!\!\int_0^\infty
  e^{-\mu(s_1+s_2)}
  \rho(|s_2-s_1|)
  \,d s_1\,d s_2 \notag\\
  &\qquad=
  2\int_0^\infty e^{-2\mu s_1}\,d s_1
  \int_0^\infty e^{-\mu h}\rho(h)\,d h \notag\\
  &\qquad=
  \frac{1}{\mu}
  \int_0^\infty e^{-\mu h}\rho(h)\,d h\,.
  \label{eq:layered-exponential-integral-reduction}
\end{align}
Substituting this reduction into equation~\eqref{eq:general-fano} gives equation~\eqref{eq:poisson-fano}.
\end{proof}

\begin{proof}[Derivation of equation~\eqref{eq:static-ceiling}]
The law of total variance gives $\Vbb[n]\geq\Ebb[n]$, which yields the lower bound. If $0<\Vbb[\lambda]<\infty$, the Cauchy--Schwarz inequality gives
\begin{align*}
  \left|\operatorname{Cov}[\lambda(t+h),\lambda(t)]\right|
  &\leq
  \sqrt{\Vbb[\lambda(t+h)]\Vbb[\lambda(t)]}
  =\Vbb[\lambda],
\end{align*}
\noindent where stationarity gives the final equality. Thus $|\rho(h)|\leq1$, and
\begin{equation}
  \int_0^\infty e^{-\mu h}\rho(h)\,d h
  \leq \int_0^\infty e^{-\mu h}\,d h
  = \frac{1}{\mu}\,.
\end{equation}
\noindent Combining this inequality with $\Ebb[n]=\Ebb[\lambda]/\mu$ gives the upper bound. If $\Vbb[\lambda]=0$, the transcription rate is constant with probability one, the rate-fluctuation contribution vanishes, and both bounds reduce to $\Fbb[n]=1$.
\end{proof}

\begin{proof}[Derivation of equation~\eqref{eq:mm1-fano}]
At stationarity, the M/M/1 state $m$ has a geometric distribution. Let us define
\begin{align}
  p_j
  &\mathrel{:=}
  \Pbb(m=j)
  =(1-r)r^j,
  \qquad j\geq0,
  \notag\\
  q
  &\mathrel{:=}
  \Ebb[m]
  =\frac{r}{1-r},
  \qquad
  \Vbb[m]
  =\frac{r}{(1-r)^2},
  \label{eq:mm1-appendix-a-stationary-statistics}
\end{align}
and also set
\begin{equation}
  \varepsilon
  \mathrel{:=}
  \frac{\mu}{m^{\downarrow}}
  =k(1-r).
  \label{eq:mm1-appendix-a-epsilon}
\end{equation}

\noindent To evaluate the exponentially weighted autocovariance in equation~\eqref{eq:poisson-fano}, we condition on the initial queue state $m(0)=j$.  For this initial state, $u_j$ accumulates the expected deviation of $m(h)$ from its stationary mean $q$, with the deviation at each future lag $h$ weighted by $e^{-\mu h}$. This weight is the probability that an mRNA survives for at least the additional time $h$: 
\begin{align}
  u_j
  \mathrel{:=}
  \int_0^\infty e^{-\mu h} \,
  \Ebb[m(h)-q\mid m(0)=j]\,d h,
  \notag\\
  v_j
  \mathrel{:=}
  m^{\downarrow}u_j.
  \label{eq:mm1-appendix-a-response}
\end{align}
\noindent
The backward equations of the birth--death process are
\begin{subequations}
\label{eq:mm1-appendix-a-recurrence}
\begin{align}
  (\varepsilon+r)v_0-rv_1
  &=-q,
  \label{eq:mm1-appendix-a-boundary}\\
  (1+r+\varepsilon)v_j-rv_{j+1}-v_{j-1}
  &=j-q,
  \qquad j\geq1,
  \label{eq:mm1-appendix-a-interior}
\end{align}
\end{subequations}
where the first line accounts for the absence of a downward transition at $m=0$. Defining the variables
\begin{equation}
  a
  \mathrel{:=}
  1+r+\varepsilon,
  \qquad
  z
  \mathrel{:=}
  \frac{2}{a+\sqrt{a^2-4r}},
  \label{eq:mm1-appendix-a-root}
\end{equation}
then $0<z<1$ and $rz^2-az+1=0$. Solving the recurrence and imposing its boundary equation gives
\begin{equation}
  v_j
  =
  \frac{j-q}{\varepsilon}
  -\frac{1-r}{\varepsilon^2}
  +\frac{z^{j+1}}{\varepsilon(1-z)}.
  \label{eq:mm1-appendix-a-response-solution}
\end{equation}
Direct substitution verifies both equations in~\eqref{eq:mm1-appendix-a-recurrence}.

The geometric distribution gives
\begin{equation}
  \Ebb[z^m]
  =\frac{1-r}{1-rz},
  \qquad
  \Ebb[(m-q)z^m]
  =-\frac{r(1-z)}{(1-rz)^2}.
  \label{eq:mm1-appendix-a-geometric-identities}
\end{equation}
Averaging the conditional response over the stationary initial state therefore yields
\begin{align}
  &\int_0^\infty e^{-\mu h}
  \bigl[\Ebb[m(h)m(0)]-q^2\bigr]\,d h
  \notag\\
  &\quad=
  \sum_{j=0}^\infty p_j(j-q)u_j
  \notag\\
  &\quad=
  \frac{1}{\mu}
  \left[
    \frac{r}{(1-r)^2}
    -\frac{rz}{(1-rz)^2}
  \right].
  \label{eq:mm1-appendix-a-covariance-transform}
\end{align}

Finally, $\lambda=\delta m$ and $\Ebb[n]=\delta q/\mu$. Equation~\eqref{eq:poisson-fano} consequently gives
\begin{align}
  \Fbb[n]-1
  &=
  \frac{\delta^2}{\mu\Ebb[n]}
  \int_0^\infty e^{-\mu h}
  \bigl[\Ebb[m(h)m(0)]-q^2\bigr]\,d h
  \notag\\
  &=
  \Ebb[n]
  \left[
    \frac{1}{r}
    -\frac{(1-r)^2z}{r(1-rz)^2}
  \right].
  \label{eq:mm1-appendix-a-fano}
\end{align}
Because $\varepsilon=k(1-r)$, equation~\eqref{eq:mm1-appendix-a-root} gives $a=1+r+k(1-r)$. This yields equation~\eqref{eq:mm1-fano}.
\end{proof}


\section{Details of Stochastic Simulations}
\label{app:simulation-details}
All simulations evolve the joint process $(\lambda(t),n(t))$. The discrete $\mathrm{M/M/1}$ model is sampled as an exact continuous-time Markov chain. The three drift--diffusion models use Euler--Maruyama steps for $\lambda(t)$ and Poisson time steps for mRNA production and degradation. After each diffusion step, the boundary rule of the corresponding model is applied to $\lambda(t)$. Stationary moments are computed as time averages after a burn-in period; ergodicity equates these long-time trajectory averages with stationary ensemble averages.

\subsection{\texorpdfstring{$\mathrm{M/M/1}$}{M/M/1} Production}
\label{app:simulation-mm1}
The simulation samples the four transitions in equation~\eqref{eq:mm1-joint-transitions}. At each event, the waiting time is exponentially distributed with a rate equal to the sum of the permitted transition rates, and the transition channel is selected in proportion to its rate. The downward transition of $m$ is omitted when $m=0$.

The initial queue state is drawn from its stationary geometric distribution, and the initial mRNA copy number is drawn conditionally as
\begin{align}
  \Pbb\bigl(m(0)=j\bigr)
  &=(1-r)r^j,
  \qquad j\geq0
  \notag\\
  n(0)\mid m(0)
  &\sim
  \operatorname{Poisson}\bigl(\delta m(0)/\mu\bigr).
\end{align}
After discarding $5\times10^4$ events, we record $N_{\mathrm{ev}}=2\times10^7$ events. If $f_j$ is the value of an observable during an inter-event interval of length $\Delta t_j$, its stationary mean is estimated by
\begin{equation}
  \widehat{\Ebb}[f]
  \mathrel{:=}
  \frac{\displaystyle\sum_{j=1}^{N_{\mathrm{ev}}}f_j\Delta t_j}
       {\displaystyle\sum_{j=1}^{N_{\mathrm{ev}}}\Delta t_j}.
\end{equation}
The simulations in Fig.~\ref{fig:mm1-model}(c) use $\mu=2$, $r=0.9$, and $k=1$, which give $m^{\downarrow}=20$ and $m^{\uparrow}=18$. We use $30$ logarithmically spaced values of $\delta$ to vary $\Ebb[n]$ from $0.1$ to $50$.

\subsection{Drift--Diffusion with a Reflecting Boundary}
\label{app:simulation-reflecting}
For the reflecting model, one Euler--Maruyama step followed by elastic reflection is
\begin{align}
  \lambda_{j+1}^{*}
  &=\lambda_j-v\Delta t+\sqrt{2D\Delta t}\,\eta_j,
  \qquad \eta_j\sim\mathcal{N}(0,1)
  \notag\\
  \lambda_{j+1}
  &=\bigl|\lambda_{j+1}^{*}\bigr|.
\end{align}
The mRNA copy number is then updated by
\begin{align}
  \nu_{+,j}
  &\sim\operatorname{Poisson}\bigl(\lambda_{j+1}\Delta t\bigr),
  \qquad
  \nu_{-,j}
  \sim\operatorname{Poisson}\bigl(\mu n_j\Delta t\bigr)
  \notag\\
  n_{j+1}
  &=\max\bigl(0,n_j+\nu_{+,j}-\nu_{-,j}\bigr).
\end{align}
The simulations in Fig.~\ref{fig:reflecting-model}(b) use $\mu=2$ and $k=2$. The main-text relations $\Ebb[n]=D/(\mu v)$ and $k=\mu D/v^2$ can be solved for $v$ and $D$ at fixed target $\Ebb[n]$ and $k$. They give
\begin{equation}
  v=\frac{\mu^2\Ebb[n]}{k},
  \qquad
  D=\frac{\mu^3\Ebb[n]^2}{k},
  \qquad
  \vartheta_p=\frac{D}{v^2}=\frac{k}{\mu}.
\end{equation}
\noindent
We use $50$ logarithmically spaced values of $\Ebb[n]$ from $0.1$ to $20$ and a time step $\Delta t=0.003k/\mu$. A burn-in of $5\max(\vartheta_p,1/\mu)$ is followed by a measurement interval of $2500\max(\vartheta_p,1/\mu)$.

\subsection{Drift--Diffusion with Periodic Boundary Conditions}
\label{app:simulation-periodic}
The Euler--Maruyama proposal is wrapped onto the interval $[0,L)$:
\begin{equation}
  \lambda_{j+1}
  =
  \bigl[
    \lambda_j-v\Delta t+\sqrt{2D\Delta t}\,\eta_j
  \bigr]\bmod L,
  \qquad
  \eta_j\sim\mathcal{N}(0,1).
\end{equation}
The copy number is updated by the same Poisson step used for the reflecting model. We initialize $\lambda(0)=L/2$, use $\Delta t=0.02$, discard a burn-in of $5\max(\vartheta_p,1/\mu)$, and measure for $6000$ time units.

The simulations in Fig.~\ref{fig:periodic-fano} use $\mu=1$ and $k=10$. We vary $L$ so that $\Ebb[n]=L/(2\mu)$ ranges from approximately $0.4$ to $40$, and set
\begin{equation}
  D=\frac{\mu L^2}{4\pi^2k}.
\end{equation}
Panel (a) uses $v=0$. Panel (b) uses $v=0.6\pi D$, for which $\vartheta_p\omega=0.6\Ebb[n]$ when $\mu=1$.

\subsection{Drift--Diffusion with First-Passage Resetting}
\label{app:simulation-fpr}
The Euler--Maruyama proposal is first reflected at $\lambda_{\min}$. If the reflected value reaches or exceeds $\lambda_{\max}=\lambda_{\min}+L$, it is reset to $\lambda_{\min}$. The copy number is then updated by the same Poisson step as above. We use
\begin{equation}
  \Delta t
  =
  \min\left(0.02,\frac{0.0008L^2}{D}\right).
\end{equation}
For the pure-diffusion simulations in Fig.~\ref{fig:fpr-fano}, the slowest-mode timescale is $\vartheta_\lambda=4L^2/(\pi^2D)$. We discard a burn-in of $5\max(\vartheta_\lambda,1/\mu)$ and record $2\times10^6$ time steps. The parameters are $D=\mu=1$, $\lambda_{\min}=0$, and $\alpha=0$. Ten logarithmically spaced values of $L$ from $1.5$ to $30$ vary $\Ebb[n]=L/3$ from $0.5$ to $10$.

\section{Derivation of Drift--Diffusion Production with Reflecting Boundary Condition}
\label{app:reflecting-derivation}
For $v>0$, equation~\eqref{eq:smoluchowski} with the reflecting boundary condition~\eqref{eq:reflecting-bc} has the stationary density
\begin{align}
  p_s(\lambda)
  &=\frac{v}{D}e^{-v\lambda/D},
  \qquad \lambda\geq0
  \notag\\
  \Ebb[\lambda]
  &=\frac{D}{v},
  \qquad
  \Vbb[\lambda]
  =\frac{D^2}{v^2}.
\end{align}
Let $C(h)$ denote the stationary autocovariance. For $h\geq0$, using the propagator~\eqref{eq:reflecting-propagator},
\begin{align}
  C(h)
  &=
  \int_0^\infty\!\!\int_0^\infty
  \lambda\lambda_0\,
  p(\lambda,h\mid\lambda_0,0)\,
  p_s(\lambda_0)\,
  d\lambda\,d\lambda_0
  -\Ebb[\lambda]^2
  \notag\\
  \rho(h)
  &=
  \frac{C(h)}{\Vbb[\lambda]}.
\end{align}
Defining $a\mathrel{:=} v^2/D$ and $k\mathrel{:=}\mu/a=\mu D/v^2$, and evaluating the integrals gives
\begin{align}
  \rho(h)
  ={}&
  \left(
    1-ah-\frac{a^2h^2}{4}
  \right)
  \operatorname{erfc}\left(\frac{\sqrt{ah}}{2}\right)
  \notag\\
  &+
  \sqrt{\frac{ah}{\pi}}
  \left(1+\frac{ah}{2}\right)
  e^{-ah/4}.
\end{align}
Its exponentially weighted integral is
\begin{equation}
  \frac{1}{\mu}
  \int_0^\infty e^{-\mu h}\rho(h)\,d h
  =
  \frac{
    2k^2-2k+\sqrt{1+4k}-1
  }{
    2k^2\mu^2
  }.
\end{equation}
Using $\Ebb[n]=D/(\mu v)$ in equation~\eqref{eq:poisson-fano} now gives
\begin{align}
  \Fbb[n]
  &=
  1+
  \frac{(D/v)^2}{\Ebb[n]}
  \left[
    \frac{1}{\mu}
    \int_0^\infty e^{-\mu h}\rho(h)\,d h
  \right]
  \notag\\
  &=
  1+\Ebb[n]
  \left[
    \frac{\sqrt{1+4k}-1}{2k^2}
    -\frac{1}{k}
    +1
  \right],
\end{align}
which yields equation~\eqref{eq:reflecting-fano}.

\section{Derivation of Drift--Diffusion Production with Periodic Boundary Conditions}
\label{app:periodic-derivation}
On the periodic interval $[0,L]$, the endpoints are identified and the process is defined on a ring. The Fourier mode $e^{2\pi i\ell\lambda/L}$ has eigenvalue $-a_\ell+i\omega_\ell$, where
\begin{equation}
  a_\ell
  \mathrel{:=}
  \frac{4\pi^2\ell^2D}{L^2},
  \qquad
  \omega_\ell
  \mathrel{:=}
  \frac{2\pi\ell v}{L}.
\end{equation}
Expanding the Dirac-delta initial condition in these modes gives
\begin{equation}
  p(\lambda,t\mid\lambda_0,0)
  =
  \frac{1}{L}
  \theta_3\left(
    \frac{\pi(\lambda-\lambda_0+vt)}{L},
    e^{-4\pi^2Dt/L^2}
  \right).
\end{equation}
Only the zero mode remains as $t\to\infty$, so the stationary density is $p_s(\lambda)=1/L$. Let us now define
\begin{equation}
  \iota_\ell
  \mathrel{:=}
  \int_0^L \lambda e^{2\pi i\ell\lambda/L}\,d\lambda
  =
  \begin{cases}
    L^2/2, & \ell=0,\\[2pt]
    L^2/(2\pi i\ell), & \ell\neq0.
  \end{cases}
\end{equation}
The stationary two-point function is therefore
\begin{equation}
  \Ebb[\lambda(t)\lambda(0)]
  =
  \frac{1}{L^2}
  \sum_{\ell=-\infty}^{\infty}
  e^{-a_\ell t}e^{i\omega_\ell t}
  \iota_\ell\iota_{-\ell}.
\end{equation}
The zero mode equals $\Ebb[\lambda]^2=L^2/4$. Combining the positive and negative modes and extending the result to all real lags gives
\begin{equation}
  C(h)
  =
  \frac{L^2}{2\pi^2}
  \sum_{\ell=1}^{\infty}
  \frac{e^{-a_\ell|h|}\cos(\omega_\ell h)}{\ell^2},
\end{equation}
which is equation~\eqref{eq:periodic-autocorr}. At $h=0$, the identity $\sum_{\ell=1}^{\infty}\ell^{-2}=\pi^2/6$ gives $\Vbb[\lambda]=L^2/12$.

For exponential degradation,
\begin{equation}
  \int_0^\infty
  e^{-\mu h}C(h)\,d h
  =
  \frac{L^2}{2\pi^2}
  \sum_{\ell=1}^{\infty}
  \frac{\mu+a_\ell}
       {\ell^2\left[(\mu+a_\ell)^2+\omega_\ell^2\right]}.
\end{equation}
Using $\Ebb[n]=L/(2\mu)$, $\vartheta_p=L^2/(4\pi^2D)$, $k=\mu\vartheta_p$, $a_\ell=\ell^2/\vartheta_p$, and $\omega_\ell=\ell\omega$ in equation~\eqref{eq:poisson-fano} gives
\begin{align}
  \Fbb[n]
  &=
  1+
  \frac{1}{\mu\Ebb[n]}
  \int_0^\infty e^{-\mu h}C(h)\,d h
  \notag\\
  &=
  1+\Ebb[n]\frac{2k}{\pi^2}
  \sum_{\ell=1}^{\infty}
  \frac{k+\ell^2}
  {\ell^2\left[
    (k+\ell^2)^2+\ell^2\vartheta_p^2\omega^2
  \right]},
  \label{eq:periodic-fano-derived}
\end{align}
which is equation~\eqref{eq:periodic-fano}. Keeping only the slowest mode and assigning it the full stationary variance gives the approximation~\eqref{eq:periodic-fano-single}. As $k\to0$, the additional term vanishes. As $k\to\infty$, the sum gives $\Ebb[n]/3$, the slow-fluctuation contribution of the uniform stationary distribution. Increasing $|\vartheta_p\omega|$ suppresses every mode through the oscillatory denominator.

\section{Derivation of Drift--Diffusion Production with First-Passage Resetting}
\label{app:fpr-derivation}
Let $L\mathrel{:=}\lambda_{\max}-\lambda_{\min}$. Within each cycle, the rate follows equation~\eqref{eq:smoluchowski}, reflects at $\lambda_{\min}$, and is absorbed at $\lambda_{\max}$. Absorbed trajectories are immediately reset to $\lambda_{\min}$. At stationarity, a constant probability current $j_s$ flows from the reset point to the absorbing boundary:
\begin{equation}
  -D\frac{d p_s}{d\lambda}-vp_s=j_s
\end{equation}
Imposing $p_s(\lambda_{\max})=0$ gives
\begin{equation}
  p_s(\lambda)
  =
  \frac{j_s}{v}
  \left[
    e^{v(\lambda_{\max}-\lambda)/D}-1
  \right]
\end{equation}
Defining
\begin{equation}
  \begin{aligned}
    \xi
    &\mathrel{:=}
    \frac{\lambda-\lambda_{\min}}{L}, \\
    \alpha
    &\mathrel{:=}
    \frac{vL}{D}, \\
    \psi(\alpha)
    &\mathrel{:=}
    \frac{e^\alpha-1}{\alpha}-1,
  \end{aligned}
\end{equation}
normalization then gives
\begin{equation}
  j_s
  =
  \frac{\alpha D}{L^2\psi(\alpha)},
  \qquad
  Lp_s(\lambda_{\min}+L\xi)
  =
  \frac{e^{\alpha(1-\xi)}-1}{\psi(\alpha)}.
\end{equation}
The second expression is the stationary density in equation~\eqref{eq:fpr-steady} written in the dimensionless coordinate $\xi$. The required integrals are
\begin{align}
  \int_0^1 e^{\alpha(1-\xi)}\,d\xi
  &=1+\psi(\alpha)
  \notag\\
  \int_0^1 \xi e^{\alpha(1-\xi)}\,d\xi
  &=\frac{\psi(\alpha)}{\alpha}
  \notag\\
  \int_0^1 \xi^2 e^{\alpha(1-\xi)}\,d\xi
  &=
  \frac{2\psi(\alpha)}{\alpha^2}
  -\frac{1}{\alpha}.
\end{align}
Consequently,
\begin{align}
  \Ebb[\xi]
  &=
  \frac{1}{\alpha}
  -\frac{1}{2\psi(\alpha)}
  \notag\\
  \Vbb[\xi]
  &=
  \frac{1}{\alpha^2}
  -\frac{1}{3\psi(\alpha)}
  -\frac{1}{4\psi(\alpha)^2},
\end{align}
which gives equation~\eqref{eq:fpr-moments}.

Now let $T(\lambda_0)$ denote the mean time required to reach $\lambda_{\max}$ from $\lambda_0$. It satisfies the backward problem
\begin{equation}
  \begin{aligned}
    DT''(\lambda_0)-vT'(\lambda_0)&=-1,\\
    T'(\lambda_{\min})&=0,\\
    T(\lambda_{\max})&=0
  \end{aligned}
\end{equation}
The process is absorbed at $\lambda_{\max}$, and reaching that boundary triggers an immediate reset to the reset point $\lambda_{\min}$. Evaluating the first-passage solution from $\lambda_{\min}$ therefore gives
\begin{equation}
  \Ebb[T_{\mathrm{cell}}]
  =
  T(\lambda_{\min})
  =
  \frac{L^2\psi(\alpha)}{\alpha D}
  =
  \frac{1}{j_s},
\end{equation}
\noindent
yielding equation~\eqref{eq:fpr-mfpt}. The result approaches $L^2/(2D)$ for pure diffusion and $L/|v|$ for strong upward drift.

For the single-cycle process without reinjection, the substitution $p(\lambda,t)=e^{-v\lambda/(2D)}\psi(\lambda,t)$ removes the first derivative from the Fokker--Planck operator and converts the spatial problem into a self-adjoint Sturm--Liouville eigenproblem. Applying the transformed absorbing and reflecting boundary conditions gives the eigenfunctions and decay rates
\begin{align}
  \phi_\zeta(\lambda)
  &\mathrel{:=}
  \sin\bigl[k_\zeta(\lambda_{\max}-\lambda)\bigr]
  \notag\\
  k_\zeta\cot(k_\zeta L)
  &=\frac{v}{2D}
  \notag\\
  E_\zeta
  &\mathrel{:=}
  Dk_\zeta^2+\frac{v^2}{4D}
\end{align}
\noindent
The absorbing condition at $\lambda_{\max}$ is automatic, and the reflecting condition at $\lambda_{\min}$ gives the equation for $k_\zeta$. For $\alpha<2$, the slowest wavenumber is real; it approaches $\pi/(2L)$ as $\alpha\to0$ and $\pi/L$ as $\alpha\to-\infty$.

If higher modes and correlations spanning reset events are neglected, the normalized autocorrelation is approximated by
\begin{equation}
  \rho(h)\approx e^{-E_0|h|},
  \qquad
  \vartheta_\lambda\mathrel{:=}\frac{1}{E_0},
  \qquad
  k_\lambda\mathrel{:=}\mu\vartheta_\lambda
\end{equation}
Substitution into equation~\eqref{eq:poisson-fano} gives
\begin{equation}
  \Fbb[n]
  \approx
  1+\Ebb[n]
  \frac{\Vbb[\lambda]}{\Ebb[\lambda]^2}
  \frac{k_\lambda}{k_\lambda+1}
\end{equation}
which is equation~\eqref{eq:fpr-fano}. Appendix~\ref{app:resolvent} retains all modes and correlations across resets.

\section{Exact Noise Transmission via the Generator Resolvent}
\label{app:resolvent}
The backward generator $\mathcal{L}$ describes how the conditional expectation of an observable $f$ changes when the initial rate is held fixed. Its evolution operator therefore satisfies
\begin{equation}
  \Ebb\bigl[f(\lambda(h))\mid\lambda(0)=\lambda\bigr]
  =
  \bigl(e^{h\mathcal{L}}f\bigr)(\lambda).
\end{equation}
\noindent Integrating this future conditional expectation against $e^{-\mu h}$ produces the resolvent $(\mu-\mathcal{L})^{-1}$. Thus the boundary-value problem yields the degradation-filtered covariance without first calculating the complete autocorrelation function. 
We define the centered rate $g(\lambda)\mathrel{:=}\lambda-\Ebb[\lambda]$ and the stationary expectation
\begin{equation}
  \Ebb_{p_s}[f]
  \mathrel{:=}
  \int_{\lambda_{\min}}^{\lambda_{\max}}
  f(\lambda)p_s(\lambda)\,d\lambda.
\end{equation}
The stationary autocovariance is then
\begin{equation}
  C(h)
  =
  \Ebb_{p_s}\left[
    g\,e^{h\mathcal{L}}g
  \right].
\end{equation}
Its exponentially weighted integral can be written as the resolvent of the generator:
\begin{align}
  J(\mu)
  &\mathrel{:=}
  \int_0^\infty e^{-\mu h}C(h)\,d h
  \notag\\
  &=
  \Ebb_{p_s}\left[
    g(\mu-\mathcal{L})^{-1}g
  \right].
\end{align}
Thus, if $\chi$ solves
\begin{equation}
  (\mu-\mathcal{L})\chi=g.
\end{equation}
then
\begin{equation}
  J(\mu)
  =
  \Ebb_{p_s}[g\chi].
\end{equation}

For the drift--diffusion process
\begin{equation}
  \mathcal{L}
  =
  -v\frac{d}{d\lambda}
  +D\frac{d^2}{d\lambda^2},
\end{equation}
the reflecting boundary at $\lambda_{\min}$ gives a Neumann condition for the backward problem. Reinjection at $\lambda_{\min}$ identifies the value of an observable at $\lambda_{\max}$ with its value immediately after reset. The resolvent problem is therefore
\begin{equation}
  \begin{aligned}
    D\chi''-v\chi'-\mu\chi
    &=
    -\bigl(\lambda-\Ebb[\lambda]\bigr),\\
    \chi'(\lambda_{\min})&=0,\\
    \chi(\lambda_{\max})&=\chi(\lambda_{\min}),
  \end{aligned}
\end{equation}
which is equation~\eqref{eq:fpr-bvp}.

Setting $y\mathrel{:=}\lambda-\lambda_{\min}\in[0,L]$, so that $\Ebb[y]=L\Ebb[\xi]$, and defining
\begin{equation}
  r_\pm
  \mathrel{:=}
  \frac{v\pm\sqrt{v^2+4D\mu}}{2D},
\end{equation}
the solution has the form
\begin{equation}
  \chi(y)
  =
  \frac{y-\Ebb[y]-v/\mu}{\mu}
  +c_+e^{r_+y}
  +c_-e^{r_-y}.
\end{equation}
The two boundary conditions determine $c_+$ and $c_-$ through
\begin{equation}
  \begin{pmatrix}
    r_+ & r_-\\[2pt]
    e^{r_+L}-1 & e^{r_-L}-1
  \end{pmatrix}
  \begin{pmatrix}
    c_+\\
    c_-
  \end{pmatrix}
  =
  \begin{pmatrix}
    -1/\mu\\[2pt]
    -L/\mu
  \end{pmatrix}.
  \label{eq:resolvent-coefficients}
\end{equation}
Finally, equation~\eqref{eq:fpr-fano-j} becomes
\begin{equation}
  \Fbb[n]
  =
  1+
  \frac{1}{\Ebb[\lambda]}
  \int_0^L
  \bigl(y-\Ebb[y]\bigr)
  \chi(y)
  p_s(\lambda_{\min}+y)\,d y.
\end{equation}
For $v=0$ and $\lambda_{\min}=0$, one has $r_\pm=\pm\sqrt{\mu/D}$, $\Ebb[y]=L/3$, and $p_s(y)=2(L-y)/L^2$. Substitution gives equation~\eqref{eq:fpr-closed-a0}. The exact resolvent retains the complete reset dynamics; the approximation~\eqref{eq:fpr-fano} instead replaces its autocorrelation by a single exponential.

\bibliography{references}

\end{document}